\documentclass[11pt,reqno]{article}   
\usepackage{fullpage}

\usepackage{microtype}
\usepackage[unicode=true]{hyperref}
\usepackage[british]{babel}
\usepackage{amsmath}
\usepackage{cite}
\usepackage{amsfonts}
\usepackage{amssymb}
\usepackage{amsthm}
\usepackage{cleveref}  
\usepackage{authblk}
\usepackage{mypack} 
\usepackage[pdftex]{color,graphicx}
\usepackage{adjustbox}
\usepackage{soul}
\usepackage{comment}
\usepackage{bbold}
\usepackage{tikz} 
\usetikzlibrary{decorations.pathreplacing,angles,quotes}
\usetikzlibrary{matrix,arrows,decorations.pathmorphing}
\usepackage{tikz-cd}
\usetikzlibrary{arrows}
\usetikzlibrary{decorations.markings}

\usepackage{bbold}
\usepackage{diagbox}
\usepackage{caption}
\usepackage{subcaption}
\usepackage{pdfpages} 
\usepackage{blkarray}
\usepackage{centernot}
\usepackage{mathtools}
\usepackage{stmaryrd}
\usepackage{soul}  
\usepackage{enumitem}

\usepackage{todonotes}
\newcounter{commcount}

\newcommand{\catCat}{\mathbf{Cat}}

\newcommand{\catEScen}{\mathbf{eScen}} 
\newcommand{\cateScen}{\mathbf{eScen}} 
\newcommand{\El}{\mathsf{El}}
\newcommand{\Func}{\mathsf{Fun}}
\newcommand{\EFunc}{\mathsf{EFun}}
\newcommand{\Bund}{\mathsf{Bund}} 
\newcommand{\sBund}{\mathsf{sBund}}

\usepackage{todonotes}
\usepackage{xcolor}

\definecolor{turquoise}{rgb}{0.2, 0.7, 0.6}

\title{Simplicial methods in the resource theory of contextuality} 
 
\author{Aziz Kharoof\footnote{aziz.kharoof@bilkent.edu.tr}}
\author{Cihan Okay\footnote{cihan.okay@bilkent.edu.tr}} 
\affil{{\small{Department of Mathematics, Bilkent University, Ankara, Turkey}}}

\date{\today}

\begin{document}

\maketitle

\begin{abstract}
We develop a resource theory of contextuality within the framework of symmetric monoidal categories, extending recent simplicial approaches to quantum contextuality. Building on the theory of simplicial distributions, which integrates homotopy-theoretic structures with probability, we introduce event scenarios as a functorial generalization of presheaf-theoretic measurement scenarios and prove their equivalence to bundle scenarios via the Grothendieck construction. We define symmetric monoidal structures on these categories and extend the distribution functor to a stochastic setting, yielding a resource theory that generalizes the presheaf-theoretic notion of simulations. Our main result characterizes convex maps between simplicial distributions in terms of non-contextual distributions on a corresponding mapping scenario, enhancing and extending prior results in categorical quantum foundations.
\end{abstract}

\tableofcontents

\section{{Introduction}}

The study of contextuality has emerged as a cornerstone of quantum foundations, revealing fundamental constraints on classical explanations of empirical data generalizing Bell inequalities. In recent years, categorical frameworks have been proposed to study contextuality in a systematic and compositional manner, most recently the simplicial framework \cite{okay2022simplicial}, which generalizes the earlier presheaf-theoretic approach \cite{abramsky2011sheaf}.
 In this work, we develop a categorical resource theory of contextuality---framed in the formalism of symmetric monoidal categories—in the sense of Coecke–Fritz–Spekkens \cite{coecke2016mathematical}, extending the approach introduced in \cite{barbosa2023closing}.

Our approach builds on the theory of simplicial distributions \cite{okay2022simplicial,kharoof2022simplicial}, which integrates simplicial sets from homotopy theory with probability theory to model spaces of distributions over measurements and outcomes. This framework extends the earlier presheaf-theoretic approach of \cite{abramsky2011sheaf} and the cohomological approach of \cite{Coho}, and further generalizes them by promoting measurements and outcomes from discrete label sets to structured spaces described by simplicial sets.
An important advantage of the simplicial  approach is its compatibility with categorical methods. In \cite{barbosa2023bundle}, the authors introduced two categories: one for bundle scenarios and one for simplicial scenarios. The former is based on simplicial complexes and acts as an intermediary between the {standard scenarios} of the presheaf-theoretic framework and the fully simplicial scenarios.
The category of bundle scenarios $\catbScen$ and the category of simplicial scenarios $\catsScen$ are connected via the nerve functor
\[
N\colon \catbScen \to \catsScen.
\]
Our first contribution in this paper is the introduction of a new category, called the category of \emph{event scenarios}, which extends the functorial viewpoint of the sheaf-theoretic approach. An event scenario is defined as a functor from the opposite of the simplex category associated to a simplicial complex (representing the measurements) to the category of sets, assigning to each simplex the set of possible outcomes, subject to certain locality conditions.  We prove that the category $\cateScen$ of event scenarios is equivalent to the category $\catbScen$ of bundle scenarios. Our proof of this result is based on a novel use of the Grothendieck construction to reinterpret categories of scenarios.

The notion of contextuality is formalized through the functors of empirical models and simplicial distributions,
\[
\bEmp\colon \catbScen \to \catConv
\qquad \text{and} \qquad
\sDist\colon \catsScen \to \catConv,
\]
which assign to each scenario the corresponding convex set of empirical models and simplicial distributions, respectively. As shown in \cite{barbosa2023bundle}, these two kinds of distributions are connected via a natural isomorphism
\[
\xi\colon \bEmp \xrightarrow{\cong} \sDist \circ N,
\]
which preserves the notion of contextuality across the two frameworks. In this work, we revisit these functors using a more streamlined construction based on a relative version of the Grothendieck construction. By interpreting bundle scenarios as event scenarios, our framework clarifies how standard presheaf-theoretic scenarios embed into this setting and makes the specialization of contextuality more transparent.

Our main motivation in this paper is to develop a resource theory for contextuality grounded in the formalism of symmetric monoidal categories, a mathematical framework for studying general resource theories introduced in \cite{coecke2016mathematical}. For instance, the resource theory of randomness arises from the distribution monad $D\colon \catSet \to \catSet$: the Kleisli category $\catSet_D$ represents the underlying \emph{process theory}, namely the category of stochastic maps, while the category of elements $\El(D)$ captures the corresponding \emph{resource theory} of probability distributions. Both categories admit natural symmetric monoidal structures used to describe parallel composition.
We extend this perspective to contextuality by introducing symmetric monoidal categories for simplicial distributions. In our setting, the process theory is given by the stochastic category of simplicial scenarios, denoted $\catsScen_D$, which arises by enlarging the morphisms of $\catsScen$ to include stochastic maps. The functor of simplicial distributions naturally extends to this stochastic category and is represented by the terminal object, encoding distributions across all scenarios. The corresponding resource theory for contextuality is then obtained by taking the category of elements of this extended functor. This construction generalizes the category of simulations introduced in the presheaf-theoretic framework \cite{barbosa2023bundle}, with morphisms capturing stochastic transformations rather than deterministic combinations of scenario morphisms.
In this former approach, the resource theory of contextuality has been studied via convex structures on empirical models and characterized through internal hom objects—so-called mapping scenarios---within the category of standard scenarios \cite{abramsky2011sheaf, coecke2016mathematical}. Our framework improves on this earlier approach by working with the  categories of event, bundle, and simplicial scenarios, where mapping scenarios can be constructed more naturally and exhibit enhanced categorical properties. In particular, our main result characterizes convex maps between simplicial distributions in terms of non-contextual distributions on the simplicial mapping scenario, thereby extending the corresponding result in the presheaf-theoretic setting to the richer framework of simplicial scenarios.

A summary of our contributions in this work are as follows:
\begin{itemize}
    \item We introduce event scenarios as a natural generalization of measurement scenarios within the presheaf-based framework of contextuality (Definition~\ref{def:funcevents}).
    
    \item A categorical equivalence between event scenarios and bundle scenarios is established in Theorem~\ref{thm:EquivCat}, providing a foundational bridge between the functorial and bundle-based perspectives.
    
    \item The empirical model and simplicial distribution functors are defined via a relative Grothendieck construction (Definitions \ref{def:empirical model functor event} and \ref{def:empirical model functor bundle}), providing a unified categorical framework for representing probabilistic data over scenarios.

    \item We endow the category of event scenarios with a symmetric monoidal structure (Definition \ref{def:tensor event scenarios}), and illustrate its use by constructing contextual empirical models through tensor composition.
    
    \item For a uniform treatment of simplicial scenarios and its stochastic extension, we define the category of scenarios over a monad with a gluing operation (Definition \ref{def:catgory of simplicial scenarios over T}).

    \item  An internal hom structure is introduced, called the mapping scenario, for both event scenarios and bundle scenarios, extending the concept from \cite{barbosa2023closing}, and a simplicial mapping scenario for simplicial scenarios used in our main result Theorem \ref{thm:contexMapp}. 
\end{itemize}

The paper is structured as follows. In Section~\ref{sec:catofscen}, we introduce the categories of event scenarios and bundle scenarios, along with their associated empirical models, and equip these categories with a symmetric monoidal structure. Section~\ref{sec:stochastic simplicial scenarios} develops the notion of stochastic simplicial scenarios via monads with a gluing operation, and defines scenarios and distributions over such monads, including a symmetric monoidal structure on this category. In Section~\ref{sec:mapping scenarios}, we introduce mapping scenarios, providing internal hom constructions for both event and bundle scenarios, as well as their simplicial set counterparts.
The appendices contain background material and technical results: Appendix~\ref{sec:gro} introduces the relative Grothendieck construction; Appendix~\ref{sec:properties of distribution monad} discusses properties of the distribution monad and includes proofs of two technical results. Appendix~\ref{sec:simplicial complexes} reviews simplicial complexes, focusing on simplicial relations and their monoidal structure. Appendix~\ref{sec:comparison results} presents comparison results related to mapping scenarios and empirical models in the setting of standard scenarios from the presheaf-theoretic perspective.

\paragraph{Acknowledgments.}
This work is supported by the Air Force Office of Scientific Research (AFOSR) under award number  FA9550-24-1-0257. The second author also acknowledges support from the Digital Horizon Europe project FoQaCiA, GA no. 101070558.

\section{Categories of scenarios}\label{sec:catofscen}

{
In this section, we introduce \emph{event scenarios}, a new class of scenarios extending the well-known notion of \emph{standard scenarios} from the sheaf-theoretic approach to contextuality \cite{abramsky2011sheaf,barbosa2023closing}. Our definition is inspired by the functorial (presheaf) perspective but generalizes it significantly, accommodating a broader range of scenarios. We define the category of event scenarios and demonstrate that it is equivalent to the category of \emph{bundle scenarios}, another generalization of standard scenarios introduced in \cite{barbosa2023bundle}. A distinctive aspect of our approach is the systematic use of the Grothendieck construction from category theory, providing a fundamental and insightful tool for building and comparing these various categories of scenarios.
}

\subsection{{Event scenarios}}
\label{sec:standard}

{In the presheaf-theoretic approach, measurements are represented by simplicial complexes. Morphisms between scenarios use simplicial relations between simplicial complexes. We will take the monadic approach developed in \cite{barbosa2023bundle} for simplicial relations, which relies on the following key construction.  

\begin{defn}\label{def:hatN}
For a simplicial complex $\Sigma$, the \emph{nerve complex} $\hat N \Sigma$ is the simplicial complex consisting of: 
\begin{itemize}
	\item The vertices given by 
the simplices of $\Sigma$, i.e., $V(\hat N\Sigma)=\Sigma$,
	\item The  
	simplices of $\hat N \Sigma$ 
	given by:
    $$
    \hat N\Sigma = \set{ \set{\sigma_1,\cdots,\sigma_k} \mid \cup_{i=1}^k\sigma_i \in \Sigma }.
    $$
\end{itemize}
\end{defn}
}

{
Simplicial complex maps $\pi\colon\Sigma' \to \hat N\Sigma$  are in one-to-one correspondence with simplcial relations. Moreover, $\hat N$ is a monad on the category $\catComp$ of simplicial complexes. We will write $\catRel$ for the associated Kleisli category.  
To a simplicial complex $\Sigma$, we will associated a category $\catC_\Sigma$ whose objects are the simplices and morphisms are given by simplex inclusions.
Then, a simplicial relation represented by $\pi$ induces a functor 
\[\overline{\pi}\colon \catC_{\Sigma'} \to \catC_{\Sigma}\]
{where a simplex} $\sigma$ {is sent} to $\cup_{\tau \in \pi(\sigma)} \tau$. For simplicity of notation, the opposite of $\bar \pi$ will also be denoted by the same symbol.
}

{Next, we introduce a subcategory that will be key in the definition of event scenarios.}  
A {\it cover} of a simplex $\sigma$ is a family of simplices $\sigma_i\subset \sigma$, where $i=1,\cdots, n$ for some $n\geq 1$, that satisfy {$\cup_{i=1}^n \sigma_i =\sigma$}. Let $\catC_{\sigma_1,\cdots,\sigma_n}$ denote the subcategory of $\catC_\Sigma$ determined by the intersections of the simplices in the cover. The inclusion functor will be denoted by
\[
\chi_{\sigma_1,\cdots,\sigma_n}\colon  \catC^{{\op}}_{\sigma_1,\cdots,\sigma_n} \to \catC^{{\op}}_\Sigma.
\]
In other words, $\chi_{\sigma_1,\cdots,\sigma_n}$ is a diagram of simplices consisting of the intersections of the simplices in the cover.

\begin{defn}\label{def:event scenario} 
A functor 
\[F\colon\catC^{\mathrm{op}}_{\Sigma} \to \catSet\]
 is called  
\emph{local}  
if, for 
every simplex $\sigma \in \Sigma$, and 
cover
$\set{\sigma_1,\cdots,\sigma_n}$ of $\sigma$, the canonical map
\[
F(\sigma) \to \lim F\circ \chi_{\sigma_1,\cdots,\sigma_n}
\]
is injective. 
\end{defn}

\begin{defn}\label{def:funcevents} 
A functor $F\colon \catC^\mathrm{op}_{\Sigma} \to \catSet$ is called an {\emph{event scenario}} if it satisfies the following 
properties:
\begin{itemize} 
    \item {\it Locality:} $F$ is local. 
    \item {\it Non-triviality:} For every $\sigma \in \catC_{\Sigma}$, the set $F(\sigma)$ is non-empty. 
    \item {\it Local surjectivity:} For every morphism $\sigma \hookrightarrow \tau$, the morphism $F(\tau) \to F(\sigma)$ is surjective.
\end{itemize}
We denote by $\EFunc(\Sigma)$ the subcategory of {the functor category} $\Func(\catC^\mathrm{op}_{\Sigma},\catSet)$ consisting of functors of events. 
\end{defn}

\begin{pro}\label{pro:fpiisoutcomes}
Let $\pi\colon \Sigma' \to \hat N \Sigma$ be a simplicial complex map and $F\colon \catC^\mathrm{op}_{\Sigma} \to \catSet$ be a local functor. 
Then, the composite 
\[
 \catC^\mathrm{op}_{\Sigma'} \xrightarrow{\bar \pi} \catC^\mathrm{op}_{\Sigma}  \xrightarrow{F} \catSet
\] 
is also a local functor. {In particular, if $F$ is an event scenario, then so is the composite $F\circ \bar \pi$.} 
\end{pro}
\begin{proof}
Given a simplex $\sigma \in \Sigma'$ and a 
cover
$\set{\sigma_1,\cdots,\sigma_n}$ of $\sigma$, 
$\set{\overline{\pi}(\sigma_1),\cdots,\overline{\pi}(\sigma_n)}$ is a 
cover
$\overline{\pi}(\sigma)$. The induced map 
$
F(\overline{\pi}(\sigma)) \to \lim F(\overline{\pi}(\chi_{\sigma_1,\cdots,\sigma_n}))
$ factors as follows:
\begin{equation}\label{eq:Flim}
\begin{tikzcd}[column sep=huge,row sep=large]
F(\overline{\pi}(\sigma))
\arrow[r,""]
\arrow[dr,""'] & \lim F(\chi_{\overline{\pi}(\sigma_1),\cdots,\overline{\pi}(\sigma_n)})  \arrow[d,""']  \\
&  \lim F(\overline{\pi}(\chi_{\sigma_1,\cdots,\sigma_n})) 
\end{tikzcd}
\end{equation}
where the vertical map 
 is injective. Moreover, since $F$ 
is local, {the horizontal map is injective.} 
{Therefore,} the induced map 
$$
F(\overline{\pi}(\sigma)) \to \lim F(\overline{\pi}(\chi_{\sigma_1,\cdots,\sigma_n}))
$$
is also injective.
\end{proof}

{Using} this result, we can 
{define}
the functor
\begin{equation}\label{eq:SFunc}
\EFunc\colon \catsRel^\op \to \catCat
\end{equation}
that sends a simplicial complex $\Sigma$ to the category  
$\EFunc(\Sigma)$, {and sends a simplicial complex map $\pi\colon \Sigma' \to \hat N \Sigma$ to the functor 
$\pi^{\ast}\colon \EFunc(\Sigma) \to \EFunc(\Sigma')$ which maps $F$ to $F \circ \overline{\pi}$.} 

We are now ready to introduce the main category {of interest} in this section. 
{Our primary tool} is the Grothendieck construction, a standard construction from category theory, which we briefly recall in Section~\ref{sec:gro}.

\begin{defn}
We define the \emph{category of event scenarios}, denoted by $\catEScen$, 
{to be}
the Grothendieck construction of the functor $\EFunc$ in (\ref{eq:SFunc}),
$$
\catEScen=
\int_{\catsRel^{\mathrm{op}}} \EFunc. 
$$
\end{defn}

{Explicitly, the category $\catEScen$ can be described as follows:}
\begin{itemize}
    \item An object in $\catEScen$ is an event scenario.
    \item A morphism in $\catEScen$ is a pair $(\pi,\alpha)\colon F \to G$, where $\pi\colon  \Sigma' \to \hat N \Sigma$ is a simplicial complex map, and $\alpha\colon F \circ \overline{\pi} \to G$ is a natural transformation.
\end{itemize}
The composition {of} $(\pi_1,\alpha)\colon F \to G$ with $(\pi_2,\beta)\colon G \to H$ is {given by}
$$
\left(\pi_1 \diamond \pi_2, \beta \circ (\Id_{\overline{\pi}_2}\star \alpha) \right)
$$
{where $\star$ denotes the whiskering operation.}
The identity on $F$ is given by $(\delta_{\Sigma},\Id_{F})$.

{An important source of examples for event scenarios come from event presheaves of standard scenarios.}

\begin{ex}\label{ex:eventpreshef}  
A \emph{standard scenario} is defined as a pair $S=(\Sigma,O)$ consisting of a simplicial complex and a family $O$ of non-empty sets $O_x$ corresponding to the outcomes of each measurement $x\in \Sigma_{0}$.
The \emph{event presheaf} is introduced as the functor 
$$
\eE_{S} \colon \catC_{\Sigma}^{\op} \to \catSet
$$
that maps each context $\sigma \in \Sigma$
to the set of joint outcomes $\prod_{x\in \sigma} O_x$,
and for an inclusion $\sigma \hookrightarrow \tau$, it maps $s\in \eE_{S}(\tau)$ to the restriction $s|_\tau\in \eE_{{S}}(\sigma)$ defined by projection.
{B}y assigning to each standard scenario $S=(\Sigma,O)$ the associated event presheaf $\eE_{S}\colon \catC_\Sigma^\op \to \catSet$, we obtain a fully faithful functor 
\[
\catScen \to \catEScen
\]
from the category of standard scenarios to the category of event scenarios.
\end{ex}

{Event scenarios generalize standard scenarios.
We now present an explicit example of an event scenario that does not arise from a standard scenario.}

\begin{ex}
Let $\Sigma$ be the simplicial complex $\set{\set{x,y},\set{x,z},\set{y,z}}$.  
{Let} $F\colon\catC^\mathrm{op}_{\Sigma} \to \catSet$ {be defined} by  
\begin{align*}
F(x)=F(y)&=F(z)=\set{0,1}\\
F(\set{x,y})&=\set{(0,0),(1,1)}\\  F(\set{y,z})&=\set{(0,0),(1,1)}\\ F(\set{x,z})&=\set{(0,1),(1,0)} 
\end{align*}
and
$$
F(x\to \set{x,y})=\pr_1, \, F(y\to \set{x,y})=\pr_2 
$$
$$
F(y\to \set{y,z})=\pr_1, \, F(z\to \set{y,z})=\pr_2 
$$
$$
F(x\to \set{x,z})=\pr_2, \, F(z\to \set{x,z})=\pr_1 
$$
where $\pr_i$ denotes the projection to the $i$-th coordinate. 
\end{ex}

\subsection{Bundle scenarios}

We begin by recalling the definition of bundle scenarios from \cite{barbosa2023bundle}. 
The \emph{star of a simplex} $\sigma \in \Sigma$ is defined by
$$
\St(\sigma) = \set{\sigma'\in \Sigma \colon\, \sigma\subset \sigma'}.
$$

\begin{defn}\label{def:BS}
A \emph{bundle scenario} is a map $f \colon \Gamma \to \Sigma$ of simplicial complexes that satisfies: 
\begin{itemize}
\item \emph{Surjectivity:} The map $f$ is surjective on simplices.
\item \emph{Local surjectivity:} For each $\gamma\in \Gamma$, the map $f|_{\St(\gamma)} \colon \St(\gamma) \to \St(f(\gamma))$ is surjective.
\item \emph{Discrete over vertices:} 
For distinct vertices $x,y \in {V(\Gamma)}$, if $f(x)=f(y)$, then $\{x,y\} \notin \Gamma$.
\end{itemize} 
Given a simplicial complex $\Sigma$, we define the category $\Bund(\Sigma)$
whose objects are bundle scenarios over $\Sigma$. 
A morphism $\alpha\colon f \to g$ in this category is a commutative diagram of the form
\begin{equation}\label{eq:alphaftog}
\begin{tikzcd}[column sep=huge,row sep=large]
\Gamma
\arrow[rr,"\alpha"]
\arrow[dr,"f"'] && \Gamma'
\arrow[dl,"g"] \\
&  \Sigma &  
\end{tikzcd}
\end{equation}
\end{defn}

{By Lemma~\ref{lem:lempi1pi2}, we obtain} a functor
\[
\Bund \colon \catsRel^{\mathrm{op}} \to \catCat.
\]
{Using this functor, we can} define the category of bundle scenarios \cite[Definition~3.7]{barbosa2023bundle} via the Grothendieck construction:

\begin{defn}\label{sec:category of bundle scenarios}
The \emph{category of bundle scenarios} $\catbScen$ is defined as the Grothendieck construction of $\Bund$:
\[
\catbScen = \int_{\catsRel^{\mathrm{op}}} \Bund.
\]
\end{defn}

The category $\catbScen$ {is} extensively studied in \cite{barbosa2023bundle}.

Our next goal is to establish an equivalence between the category of event scenarios and the category of bundle scenarios. For this purpose, we use the category-of-elements construction, a particular instance of the Grothendieck construction, as recalled in Section~\ref{sec:gro}. Specifically, we consider functors of the form $F\colon \catC_\Sigma^\op \to \catSet$, where the domain is the category associated with a simplicial complex $\Sigma$. We denote the resulting category of elements by $\catC_F$ and the canonical projection by
\[
\El(F)\colon \catC_F\to \catC_\Sigma.
\]
Given an inclusion map $i\colon \tau \hookrightarrow \sigma$, we write $s|_{\tau}$ for the image $F(i)(s)$.

\begin{pro}\label{pro:ElFsimcomp}
For a local functor 
$F\colon\catC_{\Sigma}^\mathrm{op} \to \catSet$, the category $\catC_F$ of elements 
forms
a simplicial complex, and the projection $\El(F)$ is a simplicial complex map.  
\end{pro}
\begin{proof} 
Let $(\sigma,s)$ and $(\tau,t)$ be objects in $\catC_F$. We say that $(\tau,t) \subset (\sigma,s)$ if there exists a morphism of $\catC_{F}$ from 
$(\sigma,s)$ to $(\tau,t)$, meaning {that} $\tau \subset \sigma$ and $s|_{\tau}=t$. This defines a partial order on the objects of $\catC_F$. Moreover, the projection map $\catC_F\to \catC$ preserves the partial order. Given an object $(\sigma,s)$ of $\catC_{F}$, since $F$ is local,
for any 
cover
$\set{\sigma_1,\cdots,\sigma_n}$ of $\sigma$, the object $(\sigma,s)$ is uniquely determined by its restrictions $(\sigma_{i},s|_{\sigma_i})$ where $1\leq i \leq n$. In particular, the vertices of 
$\catC_F$ will be of the form $(x,s)$, where $x \in {V(\Sigma)}$ and $s\in F(x)$, and every simplex $(\sigma,s)$ is determined by its vertices.
\end{proof}

When the functor $F$ is local, {we denote by $\Gamma_F$ the simplicial complex whose associated category is $\catC_F$.}


\begin{pro}\label{pro:ElEfunBun}
Sending a functor $F\colon\catC^\mathrm{op}_{\Sigma} \to \catSet$ to the simplicial complex map $\El(F)$ 
gives an equivalence of categories
\begin{equation}\label{eq:ElEfunBun}
\El_{\Sigma}\colon \EFunc(\Sigma) \to \Bund(\Sigma).
\end{equation}    
\end{pro}
\begin{proof}
{The proof leverages the well-known equivalence between functor categories with a fixed domain and discrete fibrations over that domain, established through the Grothendieck construction; see, e.g., \cite[Theorem~4.3]{haderi2024operadic}.}

First, we verify that the functor in (\ref{eq:ElEfunBun}) is well-defined. By Proposition \ref{pro:ElFsimcomp}, the projection $\El(F)$ can be seen as a simplicial complex map $\Gamma_F\to \Sigma$.
This map is clearly discrete over vertices, and for each $\sigma \in \Sigma$, since the set $F(\sigma)$ is non-empty, there exists
$(\sigma,s) \in \Gamma_F$ with $\El(F)(\sigma,s)=\sigma$. To prove that $\El(F)$ is locally surjective, consider 
$(\sigma',t) \in \Gamma_F$ and $\sigma' \subset \sigma$. The surjectivity of the induced map $F(\sigma) \to F(\sigma')$ ensures the existence of 
$s \in F(\sigma)$ such that $s|_{\sigma'}=t$. Hence, $(\sigma',t)\subset (\sigma,s)$ and $\El(F)(\sigma,s)=\sigma$.

Next, we construct a functor 
\begin{equation}\label{eq:thefunctorS}
S_{\Sigma} \colon  \Bund(\Sigma) \to \EFunc(\Sigma).
\end{equation}
Given a bundle scenario $f \colon \Gamma \to \Sigma$, for simplices 
$\sigma' \subset \sigma \in \Sigma$ and $\gamma\in \Gamma$ such that $f(\gamma)=\sigma$, there exists a simplex $\gamma'\subset \gamma$ with $f(\gamma')=\sigma'$. Since $f$ is discrete over vertices, such $\gamma'$ is unique. This defines a map
\begin{equation}\label{eq:rsigmasigma}
r_{\sigma,\sigma'} \colon f^{-1}(\sigma) \to f^{-1}(\sigma')
\end{equation}
which sends $\gamma$ to the corresponding unique simplex $\gamma'$. We then define
$$
S_{\Sigma}(f)(\sigma)=f^{-1}(\sigma) \;\; \text{ and }\;\; S_{\Sigma}(f)(\sigma' \hookrightarrow \sigma)=r_{\sigma,\sigma'}.
$$
Moreover, a morphism $\alpha\colon f \to g$ induces a natural transformation 
$S_{\Sigma}(f)\to S_{\Sigma}(g)$.
We now verify that $S_{\Sigma}(f)$ is an event scenario. For every $\sigma \in \Sigma$, the set $f^{-1}(\sigma)\neq \emptyset$ since $f$ is surjective. 
To show that $r_{\sigma,\sigma'}$ is surjective, consider $\gamma \in f^{-1}(\sigma')$, i.e.,
$f(\gamma)=\sigma'$. 
The map $f|_{\St(\gamma)} \colon \St(\gamma) \to \St(f(\gamma))$ is surjective, and since $\sigma \in \St(f(\gamma))$, there exists 
$\gamma \subset \tau$ such that $f(\tau)=\sigma$. Clearly, $r_{\sigma,\sigma'}(\tau)=\gamma$. Finally, for every 
cover
$\set{\sigma_1,\cdots,\sigma_n}$ of $\sigma$, the diagram $S_{\Sigma}(f)(\chi_{\sigma_1,\cdots,\sigma_n})$ is in $\catC_{\Gamma}$, so the canonical map 
$S_{\Sigma}(f)(\sigma)\to \lim S_{\Sigma}(f)(\chi_{\sigma_1,\cdots,\sigma_n})$ is injective.

Given an event scenario $F\colon \catC_{\Sigma} \to \catSet$ and $\sigma \in \Sigma$, we have
$$
S_{\Sigma}(\El_{\Sigma}(F))(\sigma)=(\El_{\Sigma})^{-1}(\sigma)=\set{(\sigma,s)|~s\in F(\sigma)} \cong F(\sigma).
$$
For an inclusion $\sigma' \subset \sigma$, the map $S_{\Sigma}(\El_{\Sigma}(F))(i\colon\sigma' \hookrightarrow \sigma)$ sends 
$(\sigma,s)$ to $(\sigma',F(i)(s))$. This establishes a natural isomorphism
$$
S_{\Sigma}\circ \El_{\Sigma} \cong \Id_{\EFunc(\catC^\mathrm{op}_{\Sigma},\catSet)}.
$$
Conversely, for a bundle scenario $f\colon \Gamma \to \Sigma$, we obtain a bundle scenario  
$
\El_{\Sigma}(S_{\Sigma}(f))\colon \Gamma_{S_{\Sigma}(f)} \to \Sigma
$
where 
$$
\Gamma_{S_{\Sigma}(f)}=\set{(\sigma,s)|~s\in S_{\Sigma}(f)(\sigma)}=\set{(\sigma,s)|~s\in f^{-1}(\sigma)}\cong \Gamma
$$
and $\El_{\Sigma}(S_{\Sigma}(f))(\sigma,s)=\sigma$. For a morphism $\alpha\colon f \to g$, we have 
$\El_{\Sigma}(S_{\Sigma}(\alpha))(\sigma,s)=(\sigma,\alpha(s))$. Thus, we also have a natural isomorphism:
$$
\El_{\Sigma} \circ S_{\Sigma}\cong \Id_{\mathbf{Bund_{\Sigma}}}.
$$ 
\end{proof}

\begin{pro}\label{pro:Equivpi}
A simplicial {complex} map $\pi\colon \Sigma' \to \hat N \Sigma$ induces a commutative square:
\begin{equation}\label{eq:pioverlinepi}
\begin{tikzcd}[column sep=huge,row sep=large]
\EFunc(\Sigma) \arrow[d,"\overline{\pi}^{\ast}"]
\arrow[r,"\El_{\Sigma}"] & \Bund(\Sigma) \arrow[d,"\pi^{\ast}"]\\
\EFunc(\Sigma') 
\arrow[r,"\El_{\Sigma'}"] & \Bund(\Sigma')
\end{tikzcd}
\end{equation}
\end{pro}
\begin{proof}
Given an event scenario $F\colon\catC^{\mathrm{op}} \to \catSet$, we prove the following square is a pullback square:
\begin{equation}
\begin{tikzcd}[column sep=huge,row sep=large]
\hat N\Gamma_{F} \arrow[d,"\hat N \El(F)"] &\arrow[l,"r_F"']  \Gamma_{F\circ \overline{\pi}}
\arrow[d,"\El(F\circ \overline{\pi})"]
 \\
\hat N\Sigma  &\arrow[l,"\pi"]  \Sigma'
\end{tikzcd}
\end{equation}
Here, the top horizontal map $r_F$ acts as follows: If $\sigma' \in \Sigma'$ such that 
$\pi(\sigma')=\set{\sigma_1',\cdots,\sigma_n'}$ and $s\in F(\overline{\pi}(\sigma'))$, then 
$$
r_F(\sigma',s)=\set{(\sigma_1,s|_{\sigma_1}),\cdots,(\sigma_n,s|_{\sigma_n})}.
$$ 
Given $\set{(\tau_1,s_1),\cdots,(\tau_m,s_m)} \in \hat N\Gamma_{F}$ 
and $\tau' \in \Sigma'$ such that $\set{\tau_1,\cdots,\tau_m}=\pi(\tau')$. Since the union $\cup_{i=1}^m (\tau_i,s_i)$ belongs to $\Gamma_F$, there is a unique $s \in F(\cup_{i=1}^n \tau_i)=F(\overline{\pi}(\tau'))$ such that $s|_{\tau_i}=s_i$ for every $1\leq i \leq m$. 
Thus, $(\tau',s)$ is the unique element in $\Gamma_{F\circ \overline{\pi}}$ satisfying 
$$
r_F(\tau',s)=\set{(\tau_1,s_1),\cdots,(\tau_m,s_m)} \;\; \text{and} \;\; \El(F\circ \overline{\pi})(\tau',s)=\tau'.
$$
So far, we have proved that 
$\El_{\Sigma'}\circ \overline{\pi}_{\ast}=\pi_{\ast}\circ \El_{\Sigma}$ in terms of object. This equality also holds for morphisms, since for any natural transformation $\alpha$ from $F$ to $G$ in $\EFunc(\catC^\mathrm{op}_{\Sigma},\catSet)$, the following diagram commutes:
$$
\begin{tikzcd}[column sep=huge,row sep=large]
\hat N \Gamma_F  \arrow[d,"\hat N {\alpha}_\ast"]&    \Gamma_{F\circ \overline{\pi}}
 \arrow[l,"r_F"]  \arrow[d,"{(\Id_{\overline{\pi}}\ast \alpha)_\ast}"]
 \\
\hat N \Gamma_G   &  \arrow[l,"r_G"] \Gamma_{G \circ \overline{\pi}}   
\end{tikzcd}
$$
\end{proof}

By Propositions \ref{pro:ElEfunBun} and \ref{pro:Equivpi}, we have a natural isomorphism of functors 
$$
\El \colon \EFunc \to \Bund.
$$ 
Applying the  
Grothendieck construction
we obtain {the desired} equivalence.

\begin{thm}\label{thm:EquivCat}
The category of elements functor induces an equivalence between the category of event scenarios and the category of bundle scenarios: 
$$
\int_{\catsRel^{\mathrm{op}}} \El \colon \catEScen \xrightarrow{\simeq}  \catbScen.
$$
\end{thm}

\subsection{Empirical models}
 
{Empirical models are fundamental objects within the sheaf-theoretic framework for contextuality \cite{abramsky2011sheaf}, describing families of probability distributions that satisfy a key constraint known as non-signaling.} In \cite{barbosa2023bundle}, the authors of the present paper generalized empirical models to bundle scenarios over simplicial complexes. In this section, we {extend empirical models further to event scenarios} and present a novel construction of the associated functor for bundle scenarios using the Grothendieck construction. {By virtue of the equivalence established in Theorem~\ref{thm:EquivCat}, the empirical models obtained by these constructions can be directly compared.}

{We begin by recalling the definition of the distribution monad and convex sets.}

\begin{defn}
The \emph{distribution monad} (see \cite[Section VI]{mac2013categories}) is a functor $D\colon  \catSet \to \catSet$ defined as follows:
\begin{itemize}
    \item For a set $X$, the set $D(X)$ of distributions on $X$ is defined to be
$$
D(X)=\set{P\colon X \to [0,\infty) \mid ~ |\set{x\in X:~P(x)\neq 0}|<\infty \,\, \text{and} \,\,\sum_{x\in X}P(x)=1}.    
$$
\item For a map $f\colon X\to Y$, the map $D(f)\colon  D(X) \to D(Y)$ is defined by
$$
P\mapsto \left( y\mapsto \sum_{x\in f^{-1}(y)} P(x) \right). 
$$
\end{itemize}
The unit of the distribution monad, $\delta_X\colon  X \xhookrightarrow{} D(X)$, sends $x\in X$ to the delta distribution $\delta^x$,
		$$
		\delta^x(x') = \left\lbrace
		\begin{array}{ll}
			1 & x'=x\\
			0 & \text{otherwise.}
		\end{array}
		\right.
		$$
\end{defn}

A \emph{convex set} consists of a set $X$ equipped with a ternary operation 
\[
\Span{-,-,-}\colon [0,1]\times X\times X \to X
\]
satisfying certain axioms that generalize convex combinations (see \cite[Definition~3]{jacobs2010convexity}). We adopt the notation $t x +(1-t)y$ instead of $\Span{t,x,y}$ for $t\in[0,1]$ and $x,y \in X$. Equivalently, a convex set can be characterized as an algebra over the distribution monad $D$ (see \cite[Theorem~4]{jacobs2010convexity}). 
We denote the category of convex sets by $\catConv$. There is an adjunction
\begin{equation}\label{eq:Set adjunction Conv}
D:\catSet \adjoint \catConv :U,
\end{equation}
where $U$ is the forgetful functor and $D$ sends a set $X$ to the free convex set $D(X)$.

For any functor $F\colon \catC \to \catSet$, composing with the distribution monad $D$ produces a functor $D_{\ast}(F)=D \circ F\colon  \catC \to \catConv$ that takes values in convex sets. 
According to \cite[Theorem 5.6.5 (i)]{riehl2017category}, the category of convex sets is a complete category, so we have 
the composite functor
$$
\Emp_\Sigma\colon\EFunc(\Sigma) \xrightarrow{D_\ast}  
\Func(\catC^\mathrm{op}_\Sigma,\catConv) \xrightarrow{\lim} 
\catConv. 
$$
Additionally, for functors $G\colon \catC_1 \to \catC_2$ and $F\colon \catC_2 \to \catSet$, there is a natural map 
\[
\lim {D_*(F)} \to \lim {D_*(F \circ G)}.
\]
This induces a functor
\[
\Emp_- \colon \catsRel^{\mathrm{op}} \to \catCat{\slice}\catConv.
\]
Applying the relative Grothendieck construction introduced in Section~\ref{sec:grothendieck}, we thereby obtain the empirical model functor on the category of event scenarios.

\begin{defn}\label{def:empirical model functor event}
We define the empirical model functor {for event scenarios} to be the relative Grothendieck construction of $\Emp_-$: 
$$
\Emp=\int_{\catsRel^{\mathrm{op}}} \Emp_-\colon  \catEScen \to \catConv. 
$$
\end{defn}

{Let $F$ be an event scenario.} 
{By} the definition of limit,  
{there is a canonical map}
\begin{equation}
\Theta_{F}\colon D(\lim F) \to \lim {D_*(F)} .
\end{equation}
\begin{defn}\label{def:conEmp}
An empirical {model} $p$ on $F$ is called \emph{contextual} if it does not lie in the image of $\Theta_F$. Otherwise, it is called \emph{noncontextual}.
\end{defn}
The functors in (\ref{eq:thefunctorS}) form a natural isomorphism  $S\colon  \Bund \to \EFunc$. This isomorphism induces the following functor:
$$
\Emp_-\circ S \colon \catsRel^{\mathrm{op}} \to \catCat{\slice}\catConv.
$$
We can now recover the empirical model functor introduced in \cite{barbosa2023bundle}, {along with the associated definition of contextuality,} on the category of bundle scenarios over simplicial complexes.

\begin{defn}\label{def:empirical model functor bundle}
We define the empirical model functor on bundle scenarios 
to be the relative Grothendieck construction: 
$$
\bEmp = \int_{\catsRel^{\mathrm{op}}} \Emp_-\circ S\colon   \catbScen \to \catConv.
$$
\end{defn}

\subsection{{Monoidal structure}}\label{subsec:MonEvent}

In this section, we show that the tensor product on $\catsComp$ {and $\catRel$}, defined {in Section \ref{subsec:MonComp}}, induces a corresponding tensor product on the category $\catEScen$, thereby endowing it with the structure of a symmetric monoidal category.

\begin{defn}\label{def:FotimesG}
Given functors $F_i\colon \catC_{\Sigma_i} \to \catSet$ for $i = 1, 2$, we define their tensor product
\[
F_1 \otimes F_2 \colon \catC_{\Sigma_1 \otimes \Sigma_2} \to \catSet
\]
by setting:
\begin{itemize}
    \item $(F_1 \otimes F_2)(\sigma) = F_1({\pr_1}(\sigma)) \times F_2({\pr_2}(\sigma))$,
    \item for an inclusion $\sigma \hookrightarrow \sigma'$, 
    \[
    (F_1 \otimes F_2)(\sigma \hookrightarrow \sigma') = 
    F_1(\pr_1(\sigma) \hookrightarrow \pr_1(\sigma')) \times 
    F_2(\pr_2(\sigma) \hookrightarrow \pr_2(\sigma')).
    \]
\end{itemize}
\end{defn}

\begin{pro}
For event scenarios $F\colon \catC^{\op}_{{\Sigma_1}} \to \catSet$ and $G\colon \catC^{\op}_{{\Sigma_2}} \to \catSet$, 
{their tensor product} $F\otimes G$ is also 
an event scenario.
\end{pro}
\begin{proof}
For a simplex ${\sigma} \in \Sigma_1 \otimes \Sigma_2$, the sets $F\left(\pr_{1}({\sigma})\right)$ and $G\left(\pr_{2}({\sigma})\right)$ non-empty, since $F$ and $G$ are non-trivial. Therefore, their product $F\left(\pr_{1}(\sigma)\right) \times G\left(\pr_{2}(\sigma)\right)$ is also non-empty.  Given 
$\sigma \hookrightarrow \sigma'$ in $\Sigma_1 \otimes \Sigma_2$, the map $
F\otimes G(\sigma \hookrightarrow \sigma') 
$ is surjective since $F\left(\pr_{1}(\sigma)\hookrightarrow \pr_{2}(\sigma')\right)$ and $G\left(\pr_{2}(\sigma)\hookrightarrow \pr_{2}(\sigma')\right)$ both are surjective. To prove locality for $F\otimes G$, consider a cover $\sigma_1,\cdots,\sigma_n$ for a simplex $\sigma$ in $\Sigma_1 \otimes \Sigma_2$, we aim to show that the map 
\begin{equation}\label{eq:Fotimescanonicalmap}
(F\otimes G)(\sigma) \to \lim \left((F \otimes G) \circ \chi_{
\sigma_1,\cdots,\sigma_n}\right)
\end{equation}
is injective. Note that $\pr_{1}$ can be seen as a functor from $\catC^{\op}_{\Sigma_1 \otimes \Sigma_2}$ to $\catC^{\op}_{\Sigma_1}$, so as proved in 
Proposition \ref{pro:fpiisoutcomes}, the functor $F \circ \pr_{1} \colon \catC^{\op}_{\Sigma_1\otimes \Sigma_2}\to \catSet$ is local. Similarity, the functor 
$G \circ \pi_{\Sigma_2} \colon \catC^{\op}_{\Sigma_1\otimes \Sigma_2}\to \catSet$ is also local. 
The map in (\ref{eq:Fotimescanonicalmap}) is the product of the canonical maps:
$$
(F\circ \pr_{1})(\sigma) \to \lim \left((F \circ \pr_{1}) \circ \chi_{\sigma_1,\cdots,\sigma_n}\right) \; \; \text{and} \;\; (G\circ \pr_{2})(\sigma) \to \lim \left((G \circ \pr_{2}) \circ \chi_{\sigma_1,\cdots,\sigma_n}\right) 
$$
Since both of these component maps are injective, their product is also injective, completing the proof.
\end{proof}

\begin{defn} \label{def:tensor event scenarios}
We define a tensor product on $\catEScen$ as follows:
\begin{itemize}
    \item For event functors $F_1$ and $F_2$, we define $F_1\otimes F_2$ as in Definition \ref{def:FotimesG}.
    \item Given event scenarios $F_i\colon \catC_{\Sigma_i} \to \catSet$, $F'_i\colon \catC_{\Sigma_i'} \to \catSet$ 
    for $i\in \set{0,1}$, 
    and morphisms $(\pi_1,\alpha)\colon F_1 \to F'_1$, $(\pi_2,\beta)\colon F_2 \to F'_2$, we define 
    $$
    (\pi_1,\alpha) \otimes (\pi_2,\beta)=(\pi_1\boxtimes \pi_2,\alpha \otimes \beta)
    $$ 
where $\pi_1\boxtimes \pi_2$ is defined as in {Corollary \ref{cor:KlisMonoi},}
and for a simplex $\sigma \in \Sigma'_1 \otimes \Sigma'_2$, we set
\begin{equation}\label{eq:alphaotimesbeta}
(\alpha \otimes \beta)_\sigma=\alpha_{\pr_{1}(\sigma)}\times \beta_{\pr_{2}(\sigma)}.
\end{equation}
\end{itemize}
\end{defn}

To illustrate {how the} map in Equation (\ref{eq:alphaotimesbeta}) {works}, suppose that $\sigma = \{(x_1, y_1), \cdots, (x_n, y_n)\}$. Then:
$$
\begin{aligned}
\pr_{1}\left(\overline{\phi_{\Sigma_1,\Sigma_2}\circ (\pi_1 \otimes \pi_2)}(\sigma)\right) 
=\pr_{1}\left(\cup_{i=1}^n (\pi_1(x_i)\times \pi_2(y_i))\right)=\cup_{i=1}^n\pi_1(x_i)=\overline{\pi_1}(\pr_{1}(A)).
\end{aligned}
$$
Similarly, $\pr_{2}\left(\overline{\phi_{\Sigma_1,\Sigma_2}\circ (\pi_1 \otimes \pi_2)}(\sigma)\right) 
=\overline{\pi_2}(\pr_{2}(\sigma))$. Therefore 
{$$
\begin{aligned}
&(F_1 \otimes F_2) (\overline{\pi_1 \boxtimes \pi_2}(\sigma)) \\
&=F_1(\pr_{1}\left(\overline{\phi_{\Sigma_1,\Sigma_2}\circ (\pi_1 \otimes \pi_2)}(\sigma)\right)) 
\times F_2(\pr_{2}\left(\overline{\phi_{\Sigma_1,\Sigma_2}\circ (\pi_1 \otimes \pi_2)}(\sigma)\right)) \\
&=F_1\left(\overline{\pi_1}(\pr_{1}(\sigma))\right) \times F_2\left( \overline{\pi_2}(\pr_{2}(\sigma))\right).
\end{aligned}
$$}

{We define the \emph{one-point event scenario} $\ast\colon \catC_{\Delta^0} \to \catSet$ by sending the unique object to the singleton set.}

\begin{pro}
$(\catEScen,\otimes,\ast)$ is a symmetric monoidal category.
\end{pro}
\begin{proof}
We will show that $\otimes\colon \catEScen \times \catEScen \to \catEScen$ is a functor, as the remaining properties follow straightforwardly. For any event scenario $F\colon \catC_{\Sigma} \to \catSet$, 
by the left-hand diagram in (\ref{eq:deltamon}), we have $\delta_{\Sigma}\boxtimes \delta_{\Sigma}=\delta_{\Sigma \otimes \Sigma}$, and it is clear that $\Id_F \otimes \Id_F$. So we have 
$$
(\delta_{\Sigma},\Id_{F}) \otimes (\delta_{\Sigma},\Id_{F})=(\delta_{\Sigma\otimes \Sigma},\Id_{F \otimes F})
$$
Next, {consider} simplicial complex maps 
$$
\pi_i\colon \Sigma'_i \to \hat N \Sigma_i\;\;\text{ and }\; \;\pi'_i\colon \Sigma''_i \to \hat N \Sigma'_i\;\; \text{for} \; i \in \set{0,1},
$$
{event scenarios}
$$
F_i\colon \catC_{\Sigma_i} \to \catSet \;,\, G_i\colon \catC_{\Sigma'_i} \to \catSet \;,\,H_i\colon \catC_{\Sigma''_i} \to \catSet \;\; \text{for} \; i \in \set{0,1},
$$
and natural tranformations
$$
\alpha\colon F_1\circ \overline{\pi_1} \to G_1\;,\, \beta \colon G_1\circ \overline{\pi'_1} \to H_1 \;,\; \gamma\colon F_2\circ \overline{\pi_2} \to G_2\;,\, \epsilon \colon G_2\circ \overline{\pi'_2} \to H_2 .
$$
Then, we have
$$
\begin{aligned}
\left((\pi'_1,\beta)\circ (\pi_1,\alpha)\right) \otimes \left((\pi'_2,\epsilon)\circ (\pi_2,\gamma)\right)
&=(\pi_1 \diamond \pi'_1,\beta\circ (\Id_{\overline{\pi'_1}}\ast \alpha)) \otimes 
(\pi_2 \diamond \pi'_2,\epsilon \circ (\Id_{\overline{\pi'_2}}\ast \gamma)) \\
&=\left((\pi_1 \diamond \pi'_1) \boxtimes (\pi_2 \diamond \pi'_2), \beta\circ (\Id_{\overline{\pi'_1}}\ast \alpha) \otimes 
\epsilon \circ (\Id_{\overline{\pi'_2}}\ast \gamma)\right)
\end{aligned}
$$
On the other hand, we have
$$
\begin{aligned}
\left((\pi'_1,\beta)\otimes(\pi'_2,\epsilon)\right) \circ \left((\pi_1,\alpha)\otimes(\pi_2,\gamma)\right)
&=(\pi'_1\boxtimes \pi'_2,\beta \otimes \epsilon)\circ(\pi_1\boxtimes \pi_2,\alpha \otimes \gamma)\\
&=\left((\pi_1\boxtimes \pi_2) \diamond (\pi'_1\boxtimes \pi'_2), (\beta\otimes \epsilon) \circ ((\Id_{\overline{\pi'_1\boxtimes \pi'_2}} \ast \alpha) \otimes \gamma)\right)
\end{aligned}
$$
By Corollary \ref{cor:KlisMonoi} $(\pi_1 \diamond \pi'_1) \boxtimes (\pi_2 \diamond \pi'_2)=(\pi_1\boxtimes \pi_2) \diamond (\pi'_1\boxtimes \pi'_2)$, so it remains to prove that
$$
\beta\circ (\Id_{\overline{\pi'_1}}\ast \alpha) \otimes 
\epsilon \circ (\Id_{\overline{\pi'_2}}\ast \gamma)=(\beta\otimes \epsilon) \circ (\Id_{\overline{\pi'_1\boxtimes \pi'_2}} \ast (\alpha \otimes \gamma)).
$$ 
For this, given ${\sigma} \in \Sigma''_1\otimes \Sigma''_2$, we have 
$$
\begin{aligned}
\left( \beta\circ (\Id_{\overline{\pi'_1}}\ast \alpha) \otimes 
\epsilon \circ (\Id_{\overline{\pi'_2}}\ast \gamma)\right)_{\sigma}&=\left(\beta\circ (\Id_{\overline{\pi'_1}}\ast \alpha)\right)_{{\pr_{1}}({\sigma})} \times 
\left(\epsilon \circ (\Id_{\overline{\pi'_2}}\ast \gamma)\right)_{{\pr_{2}}({\sigma})}\\
&=\beta_{{\pr_{1}}({\sigma})} \circ \alpha_{\overline{\pi'_1}(\pr_{1}({\sigma}))} \times \epsilon_{{\pr_{2}}({\sigma})} \circ \gamma_{\overline{\pi'_2}(\pr_{2}({\sigma}))} \\
&=\left(\beta_{\pr_{1}({\sigma})} \times \epsilon_{\pr_{2}({\sigma})} \right)\circ \left(\alpha_{\overline{\pi'_1}(\pr_{1}({\sigma}))} \times \gamma_{\overline{\pi'_2}(\pr_{2}({\sigma}))}\right) \\
&=(\beta \otimes \epsilon)_{\sigma} \circ (\alpha \otimes \gamma)_{\overline{\pi'_1\otimes \pi'_2}({\sigma})} \\
&=\left( (\beta\otimes \epsilon) \circ (\Id_{\overline{\pi'_1\boxtimes \pi'_2}} \ast (\alpha \otimes \gamma)) \right)_{\sigma}
\end{aligned}
$$
\end{proof}

Since the two categories $\catEScen$ and {$\catbScen$} are equivalent, the category {$\catbScen$} also admits a symmetric monoidal structure such that the equivalence established in Theorem~\ref{thm:EquivCat} becomes a monoidal equivalence. {We leave the explicit formulation of tensor products in $\catbScen$ to the reader.}

{Next, we show that the tensor product can be used to construct contextual empirical models from event scenarios that do not individually admit any contextual empirical models.}

 \begin{ex}
{Let $S = (\Sigma, O)$ denote the following scenario:}
\[
\Sigma = \set{\set{a}, \set{b}, \set{c}, \set{a, b}, \set{b, c}},
\]
where $O_x = \ZZ_2$ for every $x \in {V(\Sigma)}$. The associated event presheaf
\[
F = \eE_S \colon \catC_{\Sigma}^{\op} \to \catSet
\]
defines an event scenario with
\[
F(\set{a}) = F(\set{b}) = F(\set{c}) = \ZZ_2, \quad \text{and} \quad F(\set{a, b}) = F(\set{b, c}) = \ZZ_2 \times \ZZ_2.
\]
{We begin by observing that $F$ does not admit a contextual empirical model.}
Using the realization construction from \cite[Section~2.4]{kharoof2022simplicial}, the scenario $(\Sigma, O)$ corresponds to the simplicial scenario 
$(\Delta^1 \vee_{\Delta^0} \Delta^1, \Delta_{\ZZ_m})$. By \cite[Theorem~2.18]{kharoof2022simplicial}, it suffices to show that every simplicial distribution on 
$(\Delta^1 \vee_{\Delta^0} \Delta^1, \Delta_{\ZZ_m})$ is noncontextual. This follows from \cite[Example~3.11]{okay2022simplicial} and \cite[Corollary~4.6]{okay2022simplicial}.
{Now take two copies of this event scenario, i.e., let $F_i = F$ for $i=1,2$.}
We will show that although each of $F_1$ and $F_2$ admits no contextual empirical model, the event scenario $F_1 \otimes F_2$, which is not induced by any standard scenario, does admit a contextual empirical model.
The simplicial complex $\Sigma_1 \otimes \Sigma_2$ is generated by the following four pyramids:
\[
\set{a_1, b_1} \times \set{a_2, b_2}, \quad 
\set{a_1, b_1} \times \set{b_2, c_2}, \quad 
\set{b_1, c_1} \times \set{a_2, b_2}, \quad 
\set{b_1, c_1} \times \set{b_2, c_2}.
\]
An empirical model on the event scenario 
\[
F_1 \otimes F_2 \colon \catC_{\Sigma_1 \otimes \Sigma_2} \to \catSet
\]
is defined by compatible probability distributions on the following sets:
\[
\begin{aligned}
&F_1(\set{a_1, b_1}) \times F_2(\set{a_2, b_2}), \quad 
F_1(\set{a_1, b_1}) \times F_2(\set{b_2, c_2}), \\
&F_1(\set{b_1, c_1}) \times F_2(\set{a_2, b_2}), \quad 
F_1(\set{b_1, c_1}) \times F_2(\set{b_2, c_2}).
\end{aligned}
\]
We now define the empirical model $p \in \Emp(F_1 \otimes F_2)$ as follows:
\begin{equation}\label{eq:pDef1}
p_{\set{a_1,b_1} \times \set{a_2,b_2}}(i_1,j_1,i_2,j_2) = \begin{cases}
\frac{1}{2} & \text{if } (i_1,j_1,i_2,j_2) \in \set{(0,0,1,0), (1,0,0,0)}, \\
0 & \text{otherwise}
\end{cases}
\end{equation}
\begin{equation}\label{eq:pDef2}
p_{\set{a_1,b_1} \times \set{b_2,c_2}}(i_1,j_1,i_2,j_2) = \begin{cases}
\frac{1}{2} & \text{if } (i_1,j_1,i_2,j_2) \in \set{(0,0,0,0), (1,0,0,1)}, \\
0 & \text{otherwise}
\end{cases}
\end{equation}
\begin{equation}\label{eq:pDef3}
p_{\set{b_1,c_1} \times \set{b_2,c_2}}(i_1,j_1,i_2,j_2) = \begin{cases}
\frac{1}{2} & \text{if } (i_1,j_1,i_2,j_2) \in \set{(0,0,0,0), (0,1,0,1)}, \\
0 & \text{otherwise}
\end{cases}
\end{equation}
\begin{equation}\label{eq:pDef4}
p_{\set{b_1,c_1} \times \set{a_2,b_2}}(i_1,j_1,i_2,j_2) = \begin{cases}
\frac{1}{2} & \text{if } (i_1,j_1,i_2,j_2) \in \set{(0,0,0,0), (0,1,1,0)}, \\
0 & \text{otherwise}
\end{cases}
\end{equation}
We identify $F(\set{a_1, b_1})$ with $\set{0,1}_{a_1} \times \set{0,1}_{b_1}$, and similarly for the other sets. Then, the limit of the diagram $\catC_{\Sigma_1 \otimes \Sigma_2}$ is:
\begin{equation}\label{eq:limitofF}
\set{0,1}_{a_1} \times \set{0,1}_{b_1} \times \set{0,1}_{c_1} \times \set{0,1}_{a_2} \times \set{0,1}_{b_2} \times \set{0,1}_{c_2}.
\end{equation}
Suppose there exists a distribution $Q$ on the set in \eqref{eq:limitofF} such that $\Theta_F(Q)$ coincides with the empirical model $p$ defined above. Then from \eqref{eq:pDef1}, we obtain:
\begin{equation}\label{eq:con1}
Q(i_1,j_1,k_1,i_2,j_2,k_2) = 0 \quad \text{if } (i_1,j_1,i_2,j_2) \notin \set{(0,0,1,0), (1,0,0,0)}.
\end{equation}
From \eqref{eq:pDef2}, we also obtain:
\begin{equation}\label{eq:con2}
Q(i_1,j_1,k_1,i_2,j_2,k_2) = 0 \quad \text{if } (i_1,j_1,j_2,k_2) \notin \set{(0,0,0,0), (1,0,0,1)}.
\end{equation}
And from \eqref{eq:pDef3}, we get:
\begin{equation}\label{eq:con3}
Q(i_1,j_1,k_1,i_2,j_2,k_2) = 0 \quad \text{if } (j_1,k_1,j_2,k_2) \notin \set{(0,0,0,0), (0,1,0,1)}.
\end{equation}
Combining \eqref{eq:con1}, \eqref{eq:con2}, and \eqref{eq:con3}, we find:
\begin{equation}\label{eq:con123}
Q(i_1,j_1,k_1,i_2,j_2,k_2) = 0 \quad \text{if } (i_1,j_1,k_1,i_2,j_2,k_2) \notin \set{(0,0,0,1,0,0), (1,0,1,0,0,1)}.
\end{equation}
However, this contradicts \eqref{eq:pDef4}, completing the proof that $p$ is contextual.
\end{ex}

\section{Stochastic simplicial scenarios}
\label{sec:stochastic simplicial scenarios}

{In this section, we introduce a stochastic extension of the category of simplicial scenarios, originally defined in \cite{barbosa2023bundle}. Our construction is based on the notion of a category of scenarios over a monad and its simplicial refinement. This approach not only streamlines the presentation but also reveals two important special cases, corresponding to different choices of the underlying monad.}

\subsection{Monads with a gluing operation}

{For monads that weakly preserve pullbacks---such as the distribution monad---every pullback diagram admits a section. In this section, we abstract the key properties of such monads that enable the construction of stochastic versions of scenarios.  
}

Let $\catC$ be a finitely complete concrete category, and $T$ be a monad  $T\colon\catC\to \catC$.

\begin{defn}\label{def:Glungm}
A monad $T$ {is said to have} a \emph{gluing {operation}} if for every pair of morphisms $f\colon X \to Z$ and $g\colon Y \to Z$ in $\catC$, there exists a morphism
\[
m=m_{f,g}\colon T(X)\times_{T(Z)} T(Y) \to T(X \times_Z Y)\]
 of $\catC$ satisfying the following properties: 
\begin{enumerate}
\item {{\it Section property:}} The map $m_{f,g}$ is a section for the canonical map  $T(X \times_Z Y)\to T(X)\times_{T(Z)} T(Y)$.
 
\item  {\it Naturality:} 
If there are morphisms $h\colon W \to Z$ and $\beta\colon Y \to T(W)$ making the following diagram commute
\begin{equation}\label{eq:naturallll}
\begin{tikzcd}[column sep=huge,row sep=large] 
  & Y \arrow[d,"g"]
\arrow[rd,"\beta"]  & \\
X \arrow[r,"f"] \arrow[rd,hook,"\delta_X"'] & Z  \arrow[rd,hook,"\delta_Z"]  
&  T(W) \arrow[d,"T(h)"] \\
&  T(X)  \arrow[r,"T(f)"']  & T(Z) 
\end{tikzcd}
\end{equation}
then the following equation holds
$$
T(\delta_Y\times \beta)\circ m_{f,g}=m_{T(f),T(h)} \circ (T(\delta_Y)\times T(\beta)).
$$
\item  
{\it Back-and-forth property:}
For every diagram of the form
\begin{equation}\label{Eq:propofGlu}
\begin{tikzcd}[column sep=huge,row sep=large]
E \arrow[d,""]& & 
 \\
X  & Y \arrow[l,""'] & Z\arrow[l,"\pi"']
\end{tikzcd}
\end{equation}
and $(p,q)\in T(E)\times_{T(X)}T(Z)$, we have 
$$
m\left(m(p,T(\pi)(q)),q\right)=m(p,q).
$$

\item {\it Unit preservation:} 
The following diagram commutes:
$$
\begin{tikzcd}
& X \times_Z Y \arrow[dl,"\delta_X\times \delta_Y"'] \arrow[dr,"\delta_{X\times_Z Y}"] & 
 \\
 T(X)\times_{T(Z)} T(Y) \arrow[rr,"m_{f,g}"] && T(X\times_Z Y)
\end{tikzcd}
$$

\item {{\it Weak multiplicativity:}}
For every $s \in T(T(X))$ and $y \in Y$ such that $T(T(f))(s)=T(T(g))(\delta^{\delta^y})$, we have
$$
m\circ(\mu_X \times \mu_Y)(s,\delta^{\delta^y}) =
\mu_{X\times_Z Y}\circ 
T(m) \circ m (s,\delta^{\delta^y})
$$
\end{enumerate}
\end{defn}

{Our canonical example is the distribution monad.}

\begin{pro}\label{pro:distribution monad gluing}
The distribution monad $D\colon \catSet \to \catSet$ has a gluing {operation} $m$ defined by
$$
m(p,q)(x,y)=\frac{p(x)q(y)}{D(f)(p)(f(x))}.
$$     
\end{pro}
\begin{proof}
{See} Section \ref{sec:properties of distribution monad}.
\end{proof}

\subsection{Scenarios over a monad} 

We begin by developing a theory of scenarios over a monad, which we extend to the simplicial setting in the following section. This framework provides a unified treatment of constructions that generalize both the slice category and the Kleisli category.

\begin{defn}\label{def:catBunddd}
For a monad $T\colon\catC\to \catC$ and an object $X$ of $\catC$, we define the category $\Bund_{T}(X)$ consisting of:
\begin{itemize}  
    \item  objects  
    given by morphisms {$f:E\to X$} of $\catC$,
    \item a morphism from $f\colon E \to X$ to $g\colon F \to X$ is a morphism $\alpha\colon E \to T(F)$ of $\catC$ that makes the following diagram commute:
\begin{equation}\label{eq:MorinBun}
      \begin{tikzcd}[column sep=huge,row sep=large]
E \arrow[d,"f"]\arrow[r,"\alpha"] & T(F) 
\arrow[d,"T(g)"]
 \\
X  \arrow[r,hook,"\delta_X"] & T(X) 
\end{tikzcd}
\end{equation}

      \item Given $f\colon E \to X$, $g\colon F \to X$, and $h\colon G \to X$,  the composition of $\alpha\colon f \to g$ with 
      $\beta\colon g \to h$ is defined using the Kleisli composition: 
      $$
      \beta \diamond \alpha= \mu_G\circ T(\beta) \circ \alpha.
      $$ 
\end{itemize}
\end{defn}
 
{The} identity morphism {on an object $f$} 
{of}
$\Bund_{T}(X)$ is given by {the unit} $\delta_E\colon E \to T(E)$
of {the monad} $T$. 
{The composite $\mu_G\circ T(\beta) \circ \alpha$ is guaranteed to be a morphism in $\Bund_{T}(X)$ as a consequence of the following commutative diagram:}
$$
  \begin{tikzcd}[column sep=huge,row sep=large]
E \arrow[d,"f"]\arrow[r,"\alpha"] & T(F) 
\arrow[d,"T(g)"] \arrow[r,"T(\beta)"] & T(T(G)) \arrow[r,"\mu_G"]  \arrow[d,"T(T(h))"] & T(G)\arrow[d,"T(h)"]
 \\
X  \arrow[r,hook,"\delta_X"] & T(X) 
\arrow[r,hook,"T(\delta_X)"] 
\arrow[rr, bend right,equal] &T(T(X))
\arrow[r,"\mu_X"] & T(X) 
\end{tikzcd}
$$

Our construction admits the following two important special cases:
\begin{itemize}
    \item If $X$ is a terminal object in $\catC$, denoted by $\ast$, then $\Bund_{T}(\ast)$ coincides with the Kleisli category $\catC_{T}$.
    \item If $T$ is the identity monad on $\catC$, i.e., $T = \Id_\catC$, then $\Bund_{\Id_\catC}(X)$ is the slice category $\catC/X$.
\end{itemize}

\begin{defn}\label{def:piast}
Given a morphism $\pi\colon Y \to X$ of $\catC$, we define: 
\begin{itemize}
    \item For an object $f\colon E \to X$ in $\Bund_{T}(X)$, the object $\pi^\ast(f)$  is defined by the pullback square:
$$
      \begin{tikzcd}[column sep=huge,row sep=large]
E \arrow[d,"f"] &\arrow[l,""]  E \times_{X} Y 
\arrow[d,"\pi^{\ast}(f)"]
 \\
X  \arrow[ru, phantom, "\llcorner", very near end]  &\arrow[l,"\pi"]  Y 
\end{tikzcd}
$$ 
Here, $\pi^{\ast}(f)$ denotes the canonical projection from the pullback, {and we note that $E \times_X Y$ coincides with $\pi^{\ast}(E)$.}

%
\item For a morphism $\alpha$ as in Diagram (\ref{eq:MorinBun}), we define $\pi^{\ast}(\alpha)=m_{g,\pi} \circ (\alpha \times \delta_Y)$ using the following Diagram \ref{dia:pi star f}. 
\end{itemize}
\end{defn}

{\small
\begin{equation}\label{dia:pi star f}
\begin{tikzcd}[column sep=huge,row sep=large]
E 
\arrow[d,"f"]\arrow[r,"\alpha"] & T(F) 
\arrow[d,"T(g)"]
 \\
X  \arrow[r,hook,"\delta_X"] & T(X) & E \times_X Y \arrow[llu,""] \arrow[r,"\alpha \times \delta_Y"] \arrow[d,"\pi^\ast(f)"] 
& T(F) \times_{T(X)} T(Y)  \arrow[llu,""] \arrow[d,"T(\pi)^{\ast}(T(g))"] \arrow[r,"m_{g,\pi}"] & T(F\times_{X} Y) \arrow[d,"T(\pi^{\ast}(g))"]
\\
&& Y \arrow[llu,"\pi"] \arrow[r,"\delta_Y"] & 
T(Y) \arrow[llu,"T(\pi)"] \arrow[r,equal]& T(Y)
\end{tikzcd}
\end{equation}
}

{In Diagram {(\ref{dia:pi star f})},} the commutativity of the right square follows from Axiom~(1) in Definition~\ref{def:Glungm}. By combining the right two squares, we conclude that $\pi^{\ast}(\alpha) \colon \pi^{\ast}(f) \to \pi^{\ast}(g)$ is indeed a morphism in $\Bund_T(Y)$.

 \begin{ex}\label{ex:piastD}
Consider the distribution monad $D\colon \catSet \to \catSet$. A morphism as in Diagram~\eqref{eq:MorinBun} corresponds to a collection of distributions $\set{\alpha_e}_{e \in E}$ on $F$ such that each distribution $\alpha_e$ is supported on the fiber $g^{-1}(f(e))$, meaning that for every $e \in E$,
\[
\supp(\alpha_e) \subset g^{-1}(f(e)).
\]
The morphism 
$
\pi^{\ast}(\alpha) \colon \pi^{\ast}(f) \to \pi^{\ast}(g)
$
is defined componentwise by
\begin{equation}\label{eq:piastalpha}
\pi^{\ast}(\alpha)_{(e,y)}(e',y') =
\begin{cases}
\alpha_{e}(e') & \text{if } y' = y, \\
0 & \text{otherwise},
\end{cases}
\end{equation}
for elements $(e,y) \in E \times_X Y$ and $(e',y') \in F \times_X Y$.
\end{ex}

\begin{pro}
Definition \ref{def:piast} specifies a functor
$\pi^\ast \colon \Bund_{T}(X) \to \Bund_{T}(Y)$.
\end{pro}
\begin{proof}
By {property} $(4)$ of Definition \ref{def:Glungm}, we have 
$
\pi^{\ast}(\delta_E)=
m_{f,\pi} \circ (\delta_E \times \delta_Y)=\delta_{E\times_X Y}
$. This shows that $\pi^\ast$ maps identity morphisms to identity morphisms. 
To show that $\pi^\ast$ preserves compositions, given objects $f\colon E \to X$, $g\colon F \to X$, $h\colon G \to X$, and morphism $\alpha \colon E \to T(F)$, $\beta \colon F \to T(G)$. We have
$$
\begin{aligned}
\pi^{\ast}(\beta \diamond \alpha)&=m \circ \left((\mu_G \circ T(\beta) \circ \alpha)\times 
 (\mu_Y\circ T(\delta_Y)\circ\delta_Y)\right) \\
&=m \circ (\mu_G \times 
 \mu_Y) \circ (T(\beta) \times 
 T(\delta_Y) )\circ (\alpha \times 
\delta_Y)\\
&=\mu_{G \times_X Y} \circ T(m) \circ m \circ (T(\beta) \times 
 T(\delta_Y) )\circ (\alpha \times 
\delta_Y) \\
&=\mu_{G \times_X Y} \circ T(m) \circ T(\beta \times \delta_Y) \circ m\circ (\alpha \times 
\delta_Y) \\
&=\mu_{G \times_X Y} \circ T\left(m \circ (\beta \times \delta_Y)\right) \circ m\circ (\alpha \times 
\delta_Y) \\
&=\mu_{G \times_X Y} \circ T\left(\pi^{\ast}(\beta)\right) \circ \pi^{\ast}(\alpha)\\
&=\pi^{\ast}(\beta) \diamond \pi^{\ast}(\alpha).
\end{aligned}
$$
In this line of equations we used the following observations: Using the naturality of $\delta$, we obtain $T(\delta_Y)\circ \delta_Y=\delta_{T(Y)}\circ \delta_Y$. So, in line three we applied
{property} $(5)$
of Definition \ref{def:Glungm}.
We also used {property} $(2)$ of Definition \ref{def:Glungm} to replace $m \circ (T(\beta) \times 
 T(\delta_Y) )$ in line four by $T(\beta \times \delta_Y) \circ m$.
\end{proof}
\begin{lem}\label{lem:pi1piast}
For an object $X$ and 
 morphisms $\pi_1\colon Y \to X$, $\pi_2\colon Z \to Y$ of $\catC$, we have the following:
\begin{enumerate}
    \item $\Id^{\ast} \colon  \Bund_{T}(X) \to \Bund_{T}(X)$ is the identity functor.
    \item $(\pi_1 \circ \pi_2)^{\ast}=(\pi_2)^{\ast} \circ (\pi_1)^{\ast}$.
\end{enumerate}
\end{lem}
\begin{proof}
Part $(1)$ is straightforward. For part $(2)$, the case for objects is clear due to the basic properties of pullbacks.
For morphisms, given $\alpha\colon E \to T(F)$ in $\Bund_{T}(X)$, we have
$$
(\pi_2)^\ast\left((\pi_1)^\ast(\alpha)\right)=(\pi_2)^\ast\left(m_{g,\pi_1}\circ(\alpha\times \delta_Y)\right)=m_{(\pi_1)^\ast(g),\pi_2}
\circ \left(\left(m_{g,\pi_1}\circ(\alpha\times \delta_Y)\right) \times \delta_Z \right).
$$
On the other hand, we have
$$
(\pi_1 \circ \pi_2)^\ast (\alpha)
=m_{g,\pi_1 \circ \pi_2} \circ (\alpha \times \delta_Z).
$$
Since $\delta_Y\circ \pi_2=T(\pi_2) \circ \delta_Z$, the equality follows from {property} $(3)$ of Definition \ref{def:Glungm}.
\end{proof}

As a consequence of Lemma~\ref{lem:pi1piast}, we obtain a functor
\[
\Bund_T \colon \catC^{\op} \to \Cat
\]
which assigns to each object $X$ the category $\Bund_T(X)$, and to each morphism $\pi\colon Y \to X$ the pullback functor $\pi^* \colon \Bund_T(X) \to \Bund_T(Y)$.

\begin{defn}\label{def:category of scenarios over T}
We define $\catScen_T$, the \emph{category of scenarios over $T$}, to be $\int_{\catC^{\mathrm{op}}} \Bund_T$, the Grothendieck construction for 
the functor $\Bund_{T}\colon \catC^{\op} \to \catCat$.    
\end{defn}

The composition of $(\pi,\alpha)\colon f \to g$ with $(\pi',\beta)\colon g \to h$, where $h:G\to Z$ is another object, is given by
$$
(\pi \circ \pi', \beta \diamond\left(m_{g,\pi'} \circ (\alpha\times \delta_Z) \right))
$$
where $Z$ is the codomain of $h$.  
The identity on $f\colon E \to X$ is {given by} $(\Id_X,\delta_{E})$.

\subsection{Simplicial scenarios over a monad}\label{sec:Simversion}

{
Next, we introduce a simplicial version of our theory. A monad $T \colon \catC \to \catC$ extends to a monad on the category $s\catC$ of simplicial objects in $\catC$ by applying $T$ degree-wise; see \cite[Proposition~2.4]{kharoof2022simplicial}. We denote the resulting functor again by $T \colon s\catC \to s\catC$.
However, the gluing operation $m$ does not extend directly, due to the failure of naturality across simplicial degrees. Nevertheless, for every pair of morphisms $f\colon X \to Z$ and $g\colon Y \to Z$ in $s\catC$, there exists a map
\[
m_{f,g} \colon T(X) \times_{T(Z)} T(Y) \to T(X \times_Z Y)
\]
which may not be a morphism in $s\catC$; in particular, it may fail to preserve the simplicial structure. Despite this, the axioms of Definition~\ref{def:Glungm} hold in each degree separately.
}

\begin{defn}\label{def:scatBunddd}
For a monad $T\colon \catC\to \catC$ and a simplicial object $X$ of $s\catC$, we define the category $\sBund_{T}(X)$ to be the category $\Bund_{T}(X)$, 
where $T$ is now considered as the monad on $s\catC$.
\end{defn}

{
\begin{rem}
It is important to note that $\sBund_T(X)$ is \emph{not} the category of simplicial objects in $\Bund_T$. Rather, it is defined as $\Bund_T(X)$ where $T$ is regarded as a monad on $s\catC$, acting degree-wise. This distinction is essential, as $\sBund_T$ reflects the structure induced by $T$ on simplicial objects in $\catC$, rather than imposing a simplicial structure on the category $\Bund_T$ itself.
\end{rem}
}

To define a simplicial version of the functor $\pi^\ast$ from Definition \ref{def:piast}, we need to add a new axiom:
\begin{enumerate}
\item[6.] {\it Gluing with deterministics:} For $f\colon X \to Z$ and $g\colon Y \to Z$ in $\catC$, the following composition:
$$
  \begin{tikzcd}[column sep=huge,row sep=large]
T(X)\times_{T(Z)}Y 
\arrow[r,hook,"\Id_{{T(X)}} \times \delta_Y"] 
\arrow[rr, bend right,"\eta_{f,g}"'] &T(X)\times_{T(Z)}T(Y)
\arrow[r,"m_{f,g}"] & T(X\times_{Z}Y) 
\end{tikzcd}
$$  
is natural {in the sense that commutative squares:}
$$
  \begin{tikzcd}[column sep=huge,row sep=large]
X\arrow[r,"f"] \arrow[d] &Z \arrow[d] & Y \arrow[r,"g"] \arrow[d] & Z \arrow[d] \\
{X'} \arrow[r,"{f'}"]  & {Z'}  & {Y'} \arrow[r,"{g'}"] & {Z'}
\end{tikzcd}
$$  
{induce the following square:}
$$
\begin{tikzcd}[column sep=huge,row sep=large]
T(X)\times_{T(Z)}Y \arrow[r,"\eta_{f,g}"] \arrow[d] &T(X\times_{Z}Y) \arrow[d] &  \\
{T(X')\times_{T(Z')}Y' } \arrow[r,"\eta_{f',g'}"]  & {T(X'\times_{Z'}Y')} 
\end{tikzcd}
$$
\end{enumerate}

 \begin{pro}\label{pro:distribution monad deterministic}
The distribution monad satisfies the gluing with deterministic axiom.      
\end{pro}
\begin{proof}
{See} Section \ref{sec:properties of distribution monad}.
\end{proof} 

\begin{lem}\label{lem:mdeltaNatural}
Suppose we have Diagram (\ref{eq:MorinBun}) in $s\catC$, and let $\pi\colon Y \to X$ be a morphism of $s\catC$.  {If $T$ satisfies the gluing with deterministics axiom, then} for $n\geq 0$, the maps 
$$
m_{g_n,\pi_n}\circ (\alpha_n \times \delta_{Y_n})\colon E_n\times_{X_n}Y_n \to T(F_n\times_{X_n}Y_n)
$$
form a simplicial {set} map.
\end{lem}
\begin{proof}
{The following commutative diagram}
$$
{\small
  \begin{tikzcd}[column sep=huge,row sep=large]
E \arrow[rrd,"\alpha"] \arrow[d,"f"'] &&& &&\\
X \arrow[rrd,hook,"\delta_X"'] &  Y \arrow[rrd,equal] \arrow[l,"\pi"] & T(F) 
\arrow[rrd,equal] \arrow[d,"T(g)"']&&& \\
&& T(X) \arrow[rrd,equal] & Y \arrow[l,"\delta_X \circ \pi"] \arrow[rrd,hook,"\delta_Y"] & T(F) \arrow[d,"T(g)"'] &\\
&&&& T(X) & T(Y) \arrow[l,"T(\pi)"]
\end{tikzcd}
}
$$  
{implies that the simplicial map $\alpha \times \delta_{Y}$ is equal to the composition $(\Id_{T(F)} \times \delta_{Y})\circ (\alpha \times \Id_{Y})$. Therefore, for every $n \geq 0$, we have 
$$
m_{g_n,\pi_n}\circ (\alpha_n \times \delta_{Y_n})=(m_{g_n,\pi_n}\circ (\Id_{T(F_n)} \times \delta_{Y_n}))\circ (\alpha_n \times \Id_{Y_n}).
$$
We get the desired result by the gluing with deterministics axiom and since $\alpha \times \delta_Y$ is a simplicial map.}
\end{proof}
\begin{defn}\label{def:spiast}
Given a morphism $\pi\colon Y \to X$ of $s\catC$, we define the functor 
\[\pi^\ast \colon \sBund_{T}(X) \to \sBund_{T}(Y)\] 
as follows:
\begin{itemize}
    \item For an object $f\colon E \to X$ of $\sBund_{T}(X)$, the map $\pi^\ast(f)_n$ is defined to be     $(\pi_n)^\ast(f_n)$.

\item For a morphism $\alpha$ of $\sBund_{T}(X)$, the map $\pi^\ast(\alpha)_n$ is defined to be $(\pi_n)^\ast(\alpha_n)$.
\end{itemize}
{See Definition \ref{def:piast}.}
\end{defn}

By Lemma \ref{lem:mdeltaNatural}, $\pi^{\ast}(\alpha)$ is well defined simplicial map.
Definition \ref{def:spiast}  induces the functor
\[
\sBund_T\colon s\catC^{\op} \to \catCat.
\]

\begin{defn}\label{def:catgory of simplicial scenarios over T}
We define the \emph{category of simplicial scenarios over $T$} to be the Grothendieck construction:
\[
\catsScen_T = \int_{s\catC^{\mathrm{op}}} \sBund_T.
\]   
\end{defn}

The fully faithful inclusion $\catC \hookrightarrow s\catC$ induces a fully faithful inclusion 
\[\catScen_T \hookrightarrow {\catsScen}_T.\] 
On the other hand, the unit $\delta\colon \Id_{\catC}\to T$ of the monad induces a fully faithful embedding
\begin{equation}\label{dia:id to T delta}
\delta\colon {\catsScen}_{{\Id}}  \hookrightarrow {\catsScen_T}.
\end{equation}

{Next, we consider an important representable functor} 
\begin{equation}\label{eq:FT} 
F_{T,-}\colon s\catC^\op \to \catCat{\slice}\catSet
\end{equation}
defined by 
sending $X$ to the functor 
\[\sBund_T(X)(\Id_X,\cdot) {\colon \sBund_T(X) \to \catSet}.\]
Under this functor, a 
{morphism}
$\pi\colon Y\to X$ {of} $s\catC$, is sent to a morphism of $\catCat{\slice}\catSet$, i.e., {the functor $\pi^{\ast}\colon \sBund_T(X) \to  \sBund_T(Y)$ and the} natural transformation determined by the map
\[
\pi^*_{\Id_X,f}\colon  \Bund_T(X)(\Id_X,f) \to \Bund_T({Y})(\Id_Y , \pi^*(f)).
\] 
{as described in the following diagram:}
$$
\begin{tikzcd}[column sep=huge,row sep=large]
\sBund_T(X) 
\arrow[rr,"\pi^{\ast}"]
\arrow[ddr,"{\sBund_T(X)(\Id_X,\cdot)}"',""{name=A,right}] && \sBund_T(Y)
\arrow[ddl,"{\sBund_T(Y)(\Id_Y,\cdot)}",""{name=B,left}] \\
 &\arrow[Rightarrow, from=A, to=B, "{\pi^*_{\Id_X,\cdot}}"]& \\
&  \catSet&  
\end{tikzcd}
$$ 
By applying the relative Grothendieck construction described in Definition \ref{def:GenGroth}, we obtain the following functor:
\begin{equation}\label{eq:intcatC}    
F_T= \int_{s\catC^{\mathrm{op}}} F_{T,-} \colon s\catScen_T \to \catSet.
\end{equation} 
Unraveling the construction, it turns out that the functor given in (\ref{eq:intcatC}) is equal to the representable functor 
$$
\catsScen_T(\Id_{\Delta^0} ,\cdot)\colon \catsScen_T \to \catSet
$$ 
where $\Delta^0$ is the simplicial object of $\catC$ with the terminal object in every degree.

\subsection{{Stochastic simplicial scenarios and distributions}}\label{sec:SimpDistFunc}

{We now specialize Definition \ref{def:catgory of simplicial scenarios over T} to the identity monad and the distribution monad to capture our main categories of interest. When $T=\Id_{\catsSet}$, the identity monad on the category of simplicial sets, we obtain the category of simplicial scenarios} 
\[
\catsScen = \catsScen_{\Id_{\catsSet}}
\]
first introduced in \cite{barbosa2023bundle}.
An object $f \colon E \to X$ in this category consists of a simplicial set $X$, representing the space of \emph{measurements}, and a simplicial set $E$, representing the space of \emph{events}. In \cite{barbosa2023bundle}, it is shown that the category of bundle scenarios embeds fully faithfully into the category of simplicial scenarios via the nerve functor
\begin{equation}\label{dia:Nerve bundle simplicial scenarios}
N \colon \catbScen \to \catsScen.
\end{equation}

{Next, we extend the category of simplicial scenarios to a stochastic version. As we will see this extension allows a simpler distribution for the functor $\sDist$ of simplicial distributions. This new category is obtained by taking $T$ in Definition \ref{def:catgory of simplicial scenarios over T} the distribution monad $D\colon\catsSet\to \catsSet$. An explicit description is given as follows:}

\begin{defn}\label{def:stochastic simplicial scenarios}
The \emph{stochastic category of simplicial scenarios}, denoted by $\catsScen_D$, is defined as follows:
\begin{itemize}
\item An object in this category is a simplicial set map $f\colon E\to X$.
\item A morphism is a pair $(\pi,\alpha):f\to g$ is given by a diagram of simplicial sets of the form
\begin{equation}\label{dia:morphisms}
\begin{tikzcd}[column sep=huge,row sep=large]
E \arrow[d,"f"] &  \pi^*(E) \arrow[l] \arrow[d] \arrow[r,"\alpha"] & D(F) \arrow[d,"D(g)"] \\
X& \arrow[l,"\pi"] \arrow[r,"\delta_Y"] Y & D(Y) 
\end{tikzcd}
\end{equation}
where the left square is a pullback and the right square commutes.
\end{itemize}
\end{defn}

%

The category of simplicial scenarios embeds into the category of stochastic simplicial scenarios via the fully faithful functor
\begin{equation}\label{eq:deltasScensScenD}
\delta \colon \catsScen \hookrightarrow \catsScen_D,
\end{equation}
{obtained by specializing Diagram (\ref{dia:id to T delta}).}

\begin{pro}\label{pro:BunConv}
Let $\pi \colon Y \to X$ be a simplicial {set} map, and let $f, g$ be objects of $\sBund_D(X)$. Then the induced map
\[
\pi^\ast = \pi^\ast_{f,g} \colon \sBund_D(X)(f, g) \to \sBund_D(Y)(\pi^\ast(f), \pi^\ast(g))
\]
is a morphism in the category $\catConv$, i.e., it preserves convex combinations.
\end{pro}

%
\begin{proof} 
The convex structure on $s\catBund_{D}(X)(f,g)$ defined as follows:
$$
(t \alpha +(1-t)\beta)_{e}(e')=t \alpha_{e}(e')+(1-t)\beta_{e}(e')
$$
where $e \in E_n$, $e' \in F_n$, and $t\in[0,1]$. Now, consider $(e,y)\in E\times_X Y$ and $(e',y')\in F\times_X Y$. By 
Equation (\ref{eq:piastalpha}), if $y'=y$, we have:
$$
\begin{aligned}
\pi^\ast(t\alpha+(1-t)\beta)_{(e,y)}(e',y')&=(t \alpha+(1-t)\beta)_e(e')
\\ 
&=t\alpha_e(e')+(1-t)\beta_e(e') \\
&=t\pi^{\ast}(\alpha)_{(e,y)}(e',y')+(1-t)\pi^{\ast}(\beta)_{(e,y)}(e',y') \\
&=\left(t\pi^{\ast}(\alpha)+(1-t)\pi^{\ast}(\beta)\right)_{(e,y)}(e',y') .
\end{aligned}
$$
If $y'\neq y$, then both sides are zero, which completes the proof.
\end{proof}

{As an immediate consequence, the functor in (\ref{eq:FT}) has its target in the category $\catCat{\slice}\catConv$, and thus we obtain the functor
\[
F_{D,-} \colon \catsSet^{\op} \to \catCat{\slice}\catConv.
\]
}

\begin{defn}\label{def:simplicial distributions functor}
The 
\emph{simplicial distribution functor} is defined to be the relative Grothendieck construction of $\sDist_-$: 
\[
\sDist = \int_{\catsSet^\op} \sDist_-\colon \catsScen_D \to \catConv .
\] 
\end{defn}

By composing this functor with the inclusion in~\eqref{eq:deltasScensScenD}, we recover the simplicial distribution functor introduced in~\cite{barbosa2023bundle}. This relationship is expressed by the following commutative diagram:
\[
\begin{tikzcd}[column sep=huge, row sep=large]
\catsScen \arrow[r,"\sDist"] \arrow[d,"\delta"',hookrightarrow] & \catConv \\
\catsScen_D \arrow[ru,"\sDist"'] 
\end{tikzcd}
\]

{Moreover, the simplicial distributions functor is related to the empirical models functor defined on bundle scenarios via the nerve functor in (\ref{dia:Nerve bundle simplicial scenarios}). As shown in \cite{barbosa2023bundle} there is a natural isomorphism}
\[
\zeta \colon {\bEmp} \to \sDist \circ N.
\]

{The notion of contextuality also extends to the simplicial set setting.} 
Given a simplicial bundle scenario $f \colon E \to X$, let $\sSect(f)$ denote the set of sections of $f$. We define the map
\begin{equation}\label{dia:Theta_f}
\Theta_f \colon D(\sSect(f)) \to \sDist(f)
\end{equation}
to be the transpose of the inclusion 
$(\delta_E)_\ast \colon \sSect(f) \hookrightarrow \sDist(f)$
under the adjunction (\ref{eq:Set adjunction Conv}).

\begin{defn}\label{def:contextuality}
A simplicial distribution $p$ on $f$ is called \emph{contextual} if it does not lie in the image of the map $\Theta_f$. 
Otherwise, $p$ is called \emph{noncontextual}.
\end{defn}

For further details, see \cite[Section 5.2]{barbosa2023bundle}.

Now, we define a tensor product on $\catsScen_D$, the stochastic category of simplicial scenarios: The monoidal product of the scenarios 
\[
f \colon E \to X \quad \text{and} \quad g \colon F \to Y
\]
is given by
\[
f \otimes g \colon E \times F \xrightarrow{f \times g} X \times Y.
\]
The monoidal unit is the identity map $\Id_{\Delta^0}$ on the $0$-simplex.

{To define the monoidal product for morphisms, we use the gluing operation for distributions on simplicial sets. That is, given simplicial sets $X$ and $Y$, we have a simplicial set map
\[
m\colon D(X)\times D(Y) \to D(X\times Y)
\]
obtained by applying the gluing operation level-wise: For distributions $p\in D(X_n)$ and $q\in D(Y_n)$, $m(p,q)(x,y)=p(x)\cdot q(y)$.}
Then, given morphisms $(\pi_1,\alpha_1)\colon f_1 \to g_1$ and $(\pi_2,\alpha_2)\colon f_2 \to g_2$ in $\catsScen_D$, we obtain the following diagram:
$$ 
\begin{tikzcd}[column sep=large ,row sep=large]
E_1 \times E_2 \arrow[d,"{f_1 \otimes f_2}"] &  (E_1\times_{X_1} Y_1) \times (E_2\times_{X_2} Y_2) \arrow[l]
\arrow[d] \arrow[r,"{\alpha_1 \times \alpha_2}"] & D(F_1)\times D(F_2) \arrow[d,"{D(g_1) \times D(g_2)}"] \arrow[r,"m"] & D(F_1\times F_2) \arrow[d,"{D(g_1\times g_2)}"]  \\
X_1 \times X_2 & \arrow[l,"{\pi_1\times \pi_2}"] \arrow[r,"{\delta_{Y_1}\times \delta_{Y_2}}"] Y_1\times Y_2 & D(Y_1)\times D(Y_2) 
\arrow[r,"m"] & D(Y_1\times Y_2) 
\end{tikzcd} 
$$
{where the first square is a pullback and the other two squares commute.}
By property (4) in Definition \ref{def:Glungm}, we conclude that 
$
\left(\pi \times \pi',\, m \circ (\alpha \times \alpha')\right)
$
is a morphism in $\catsScen_D$, which we take as the definition of the tensor product 
\begin{equation}\label{eq:tensor stochastic}
(\alpha_1,\pi_1)\otimes (\alpha_2,\pi_2):=\left(\pi \times \pi',\, m \circ (\alpha \times \alpha')\right).
\end{equation}
{Then $(\catScen_D, \otimes, \Id_{\Delta^0})$ forms a symmetric monoidal category. We leave the verification of the monoidal axioms to the reader.}

\section{Mapping scenarios}
\label{sec:mapping scenarios} 
 
%

{In this section, we introduce the mapping scenario construction for the categories $\catEScen$, $\catbScen$, and $\catsScen$. Compared to the original formulation of mapping scenarios in the sheaf-theoretic framework---which tends to be more intricate---our approach yields a cleaner and more uniform treatment across these three categories. We then establish the connections between these constructions via the functors introduced in \cite{barbosa2023bundle}, which relate the categories.}

\subsection{Mapping scenarios {of event and bundle scenarios}}\label{subsec:MapBund}
 
In this section, we define the notion of a mapping scenario within the broader framework of scenarios introduced in Section \ref{sec:catofscen}. We demonstrate that the convex set of empirical models on a mapping scenario is naturally isomorphic to the set of empirical models defined on the analogous construction given in \cite[Definition 4.1]{barbosa2023closing} for standard scenarios.

For a simplex $\sigma$ in $\Sigma$, let $\Delta_\sigma$ denote the subsimplicial complex of $\Sigma$ generated by $\sigma$, i.e., $\Delta_\sigma$ consists of non-empty subsets of $\sigma$. Given a functor $F \colon \catC_\Sigma^\op \to \catSet$, we write $F|_{\sigma}$ to denote its restriction to the subcategory $\catC_{\Delta_\sigma}^\op$.

\begin{defn}
Given functors $F\colon \catC_{\Sigma}^{\mathrm{op}}\to \catSet$ and $G\colon \catC_{\Sigma'}^{\mathrm{op}}\to \catSet$, we define 
\[[F,G] \colon \catC_{\Sigma'}^{\mathrm{op}}\to \catSet\] 
as follows:
\begin{itemize}
    \item {For $\sigma\in \Sigma'$,} we define $[F,G](\sigma)=\EFunc(F,G|_{\sigma})$.
    \item For 
{an inclusion of simplices}    
    \( j\colon \sigma \hookrightarrow \tau \), we define 
    \[[F,G](j)(\pi,\alpha)=(\pi|_{\Delta_\sigma}, \Id_{\overline{j}} \ast \alpha),\] 
    where {\( \overline{j} \colon \catC_{\Delta_\sigma}^{\mathrm{op}} \hookrightarrow \catC_{\Delta_\tau}^{\mathrm{op}} \) is the functor induced by $j$.}
\end{itemize}
\end{defn}

\begin{rem}\label{rem:NatDeter}
If $F$ is {locally} surjective, then the natural transformation $\alpha\colon F\circ \overline{\pi} \to G|_{\sigma}$ is fully determined by $\alpha_{\sigma}$. Thus, for 
$\sigma=\set{x_1,\cdots,x_n}$, an element $(\pi,\alpha) \in [F,G](\sigma)$ is determined by the simplex $\pi(\sigma)=\set{\tau_1,\cdots,\tau_n} \in 
\hat N \Sigma$ and the set map $\alpha_{\sigma}\colon F(\cup_{i=1}^n \tau_i) \to G(\sigma)$.
\end{rem}
\begin{pro}\label{pro:FGEvent}
If $G$ is   
{an event scenario,}
then $[F,G]$ is    
{also an event scenario.}
\end{pro}
\begin{proof}
First, it is clear that for every $\sigma \in \Sigma'$ we have $[F,G](\sigma) \neq \emptyset$, since $G(\sigma) \neq \emptyset$. To prove local surjectivity, consider
$j\colon \sigma \hookrightarrow \tau$ in $\catC_{\Sigma'}$ and $(\pi,\beta) \in \EFunc(F,G|_{\sigma})$. Extend  
$\pi\colon \Delta_\sigma \to \hat N{\Sigma}$ to a map $\pi' \colon \Delta_\tau \to \hat N\Sigma$. Since $G$ is locally surjective, we can (using the axiom of choice), pick for each 
$a\in G(\sigma)$ a preimage $a'\in G(\tau)$. So we also get a choice of one preimage for $a\in G(\sigma)$ in $G(\tau')$ for 
every simplex $\sigma \subset \tau' \subset \tau$. Using these preimages, we define maps $\alpha_{\tau'}$ that make the following 
diagram commute 
$$
\begin{tikzcd}[column sep=huge,row sep=large]
F(\overline{\pi'}(\tau')) \arrow[d,""] \arrow[r,"\alpha_{\tau'}"]&   G(\tau')  
\arrow[d,""]
 \\
F(\overline{\pi'}(\sigma))=F(\overline{\pi}(\sigma))  \arrow[r,"\beta_{\sigma}"]&  G(\sigma)   
\end{tikzcd}
$$
Because of the compatible choices of the preimages, this defines a natural map 
$\alpha\colon F \circ \overline{\pi'} \to G|_{\Delta_\tau}$ that satisfies $\Id_{\overline{j}}\ast \alpha=\beta$.
Finally, to prove locality, given a simplex $\sigma \in \Sigma'$ and a cover $\set{\sigma_1,\cdots,\sigma_n}$ of $\sigma$, we prove that the canonical map 
$$
L\colon [F,G](\sigma) \to \lim [F,G]\circ \chi_{\sigma_1,\cdots,\sigma_n}
$$ 
is injective. Given $(\pi,\alpha), (\pi',\beta) \in [F,G](\sigma)$, and $\left((\pi_1,\alpha_1),\cdots,(\pi_n,\alpha_n)\right) \in \lim [F,G]\circ \chi_{\sigma_1,\cdots,\sigma_n}$ such that 
$$
L(\pi,\alpha)=\left((\pi_1,\alpha_1),\cdots,(\pi_n,\alpha_n)\right)=L(\pi',\beta).
$$
First, we have $\pi|_{\sigma_i}=\pi_i=\pi'|_{\sigma_i}$ for every $1\leq i \leq n$. Since $\set{\sigma_1,\cdots,\sigma_n}$ is a cover of $\sigma$, we conclude that $\pi=\pi'$. Therefore, 
$\alpha_{\sigma}$ and $\beta_{\sigma}$ both making the following diagram commute:

$$
\begin{tikzcd}[column sep=huge,row sep=large]
F (\overline{\pi}(\sigma)) \arrow[d,""] \arrow[rr,""] & &  G(\sigma)  
\arrow[d,""]
 \\
\lim F\circ \overline{\pi}\circ \chi_{\sigma_1,\cdots,\sigma_n}  \arrow[rr,"\left((\alpha_1)_{\sigma_1}{,}\cdots {,}(\alpha_n)_{\sigma_n}\right)"]& &
\lim G\circ \chi_{\sigma_1,\cdots,\sigma_n}   
\end{tikzcd}
$$
Since $G$ is local, we {obtain} that $\alpha=\beta$.
\end{proof}

{

\begin{rem}
The mapping scenario construction for event scenarios generalizes the original sheaf-theoretic notion of mapping scenarios. A detailed comparison and proof of this generalization are provided in Section \ref{sec:comparison of mapping scenarios}.
\end{rem}

We now introduce the notion of mapping scenarios for bundle scenarios, providing an analogue of the construction defined for event scenarios.
}

\begin{defn}
Given bundle scenarios $f\colon \Gamma \to \Sigma$ and $g\colon \Gamma' \to \Sigma'$, we define the simplicial complex
\[
\Gamma(f, g) = \bigsqcup_{\sigma \in \Sigma'} \catbScen(f, g|_{g^{-1}(\Delta_{\sigma})}),
\]
where $\Delta_{\sigma}$, as before, denotes the subsimplicial complex generated by $\sigma$.
More explicitly, a simplex in $\Gamma(f, g)$ is a diagram of the form:
\begin{equation}\label{dia:morphisms}
\begin{tikzcd}[column sep=huge,row sep=large]
\hat N \Gamma \arrow[d,"\hat N f"] &  \hat N \Gamma\times_{\hat N \Sigma} \Delta_{\sigma}  \arrow[l] \arrow[d,""] \arrow[r,"\alpha"] & g^{-1}(\Delta_{\sigma}) \arrow[ld,"g|_{g^{-1}(\Delta_{\sigma})}"] \\
\hat N \Sigma& \arrow[l,"\pi=\set{\tau_1,\cdots,\tau_k}"] \Delta_{\sigma} 
\end{tikzcd}
\end{equation}
for some $\sigma \in \Sigma'$. We denote such a diagram by $(\sigma; \pi, \alpha)$.
The face relation in $\Gamma(f, g)$ is defined as follows:
\[
(\sigma; \pi, \alpha) \subset (\tilde{\sigma}; \tilde{\pi}, \beta) \quad \text{if and only if} \quad 
\sigma \subset \tilde{\sigma}, \quad \tilde{\pi}|_{\Delta_{\sigma}} = \pi, \quad \text{and} \quad 
\beta|_{\hat N \Gamma \times_{\hat N \Sigma} \Delta_{\sigma}} = \alpha.
\]
Finally, the \emph{mapping scenario} is defined as the projection
\[
[f, g] \colon \Gamma(f, g) \to \Sigma',
\]
sending $(\sigma; \pi, \alpha)$ to $\sigma$.
\end{defn}

\begin{lem}\label{lem:LhatNf}
If $f\colon \Gamma \to \Sigma$ is discrete over vertices, then for any $\set{\tau_1,\cdots,\tau_n} \in \hat N\Sigma$, there is a {bijection}
$$
L_{\set{\tau_1,\cdots,\tau_n}} \colon (\hat 
 N f)^{-1}\left(\set{\tau_1,\cdots,\tau_n}\right) \to f^{-1}(\cup_{i=1}^n \tau_i),
$$
given by mapping $\set{\gamma_1,\cdots,\gamma_k}$ to the union $\cup_{i=1}^k\gamma_i$.
\end{lem}
\begin{proof}
First, since $f$ is discrete over vertices and 
$\cup_{i=1}^k \gamma_i \in \Gamma$, we may assume that $k=n$ and 
$f(\gamma_j)=\sigma_j$ for every $1\leq j \leq n$. 
The inverse of $L_{\set{\tau_1,\cdots,\tau_n}}$ is given by sending $\gamma \in f^{-1}(\cup_{i=1}^n \tau_i)$ to 
 $$
 \set{r_{\cup_{i=1}^n\tau_i,\tau_1}(\gamma),\cdots,r_{\cup_{i=1}^n\tau_i,\tau_n}(\gamma)},
 $$
 as defined in the map (\ref{eq:rsigmasigma}). It is clear that
 $r_{\cup_{i=1}^n\tau_i,\tau_j}(\cup_{i=1}^n\gamma_i)=\gamma_j$ for every $1\leq j \leq n$. Conversely,
we obtain that $\cup_{j=1}^nr_{\cup_{i=1}^n\tau_i,\tau_j}(\gamma)=\gamma$ since $f$ is discrete over vertices.
\end{proof}

\begin{pro}\label{pro:MApEl=ElMap}
Let $F\colon \catC_{\Sigma} \to \catSet$ and $G\colon \catC_{\Sigma'} \to \catSet$ be {event scenarios}. Then there exists an isomorphism of simplicial complexes
\[
\Gamma(\El(F), \El(G)) \cong \catC_{[F, G]},
\]
which makes the following diagram commute:
\begin{equation}\label{dia:Moralphaa}
\begin{tikzcd}[column sep=huge,row sep=large]
\Gamma(\El(F),\El(G))
\arrow[rr,"\cong"]
\arrow[dr,"{[\El(F),\El(G)]}"'] && \catC_{[F,G]}  
\arrow[dl,"\El({[F,G]})"] \\
&  \Sigma' &  
\end{tikzcd}
\end{equation}
\end{pro}
\begin{proof}
Given $(\sigma;\pi,\alpha) \in \Gamma(\El(F),\El(G))$. We have a simplex 
$\sigma \in \Sigma'$, with $\pi(\sigma)=\set{\tau_1,\cdots,\tau_k} \in \hat N \Sigma$,
and a simplicial complex map $\alpha$ 
making the following diagram commute:
\begin{equation}\label{eq:NCFalphaElG}
\begin{tikzcd}[column sep=huge,row sep=large]
\pi^{\ast}(\hat N \catC_F)
\arrow[rr,"\alpha"]
\arrow[dr,"\pi^{\ast}(\El(F))"'] && \El(G)^{-1}(\Delta_{\sigma})  
\arrow[dl,""] \\
&  \Delta_{\sigma} &  
\end{tikzcd}
\end{equation}
See Diagram (\ref{dia:piastf}). By Proposition \ref{pro:ElEfunBun}, the map $\El(F)$ is locally surjective, so by \cite[Propositions A.6 and A.7]{barbosa2023bundle}, we have that the left-hand projection 
$\pi^{\ast}(\El(F))$ in (\ref{eq:NCFalphaElG}) is also locally surjective. Therefore, $\alpha$ is determined by its restriction
$\alpha|_{\pi^{\ast}(\El(F))^{-1}(\sigma)}$. Using Lemma \ref{lem:LhatNf} and the definition of $\El$, we obtain 
$$
\pi^{\ast}(\El(F))^{-1}(\sigma)\cong (\hat N  \El(F))^{-1}(\set{\tau_1,\cdots,\tau_k})\cong \El(F)^{-1}(\cup_{i=1}^k\tau_i) \cong F(\cup_{i=1}^k\tau_i).
$$
So $\alpha$ corresponds to a map $\alpha' \colon F(\cup_{i=1}^k\tau_i) \to G(\sigma)$. Then $(\pi,\alpha')$ form{s} an object 
in $[F,G](\sigma)=\EFunc(F,G|_{\sigma})$. The isomorphism in diagram (\ref{dia:Moralphaa}) is defined by mapping $(\sigma;\pi,\alpha)$ to $(\sigma,(\pi.\alpha'))$.
\end{proof}
\begin{cor}
The map $[f,g]\colon  \Gamma(f,g) \to \Sigma'$ is a bundle scenario.    
\end{cor}
\begin{proof}
Since for every bundle scenario $f\colon \Gamma \to \Sigma$ is isomorphic to $\El(S_{\Sigma}(f))$ in $\Bund(\Sigma)$, the result follows from Propositions 
\ref{pro:ElEfunBun}, \ref{pro:FGEvent}, and \ref{pro:MApEl=ElMap}.    
\end{proof}

{The mapping scenarios introduced for event and bundle scenarios do not admit an adjunction with the tensor product defined in Section~\ref{subsec:MonEvent}. As a consequence, they do not yield a closed monoidal structure with respect to that tensor product. We choose to address this limitation in the simplicial set formulation, which will be developed in the next section.
}

\subsection{{Simplicial m}apping scenario{s}}

We now construct the mapping scenario in the category $\catsScen$.

\begin{defn}
Let $f\colon E\to X$ and $g\colon F\to Y$ be simplicial scenarios.
The simplicial mapping scenario is the simplicial set map
$$
\Map(f,g)\colon  E(f,g) \to Y
$$
where the space of events is defined by
$$
E(f,g)_n = \coprod_{y\in Y_n} \catsScen(f,g_y)
$$ 
Here, $g_y\colon y^*(F)\to \Delta^n$ is the pull-back along the $n$-simplex $y\colon \Delta^n\to Y$.  
\end{defn} 

An $n$-simplex in 
$E(f,g)$ consists of $y\in Y_n$, $x \in X_n$, and the following {commutative} diagram {of simplicial sets}
%
$$
\begin{tikzcd}[column sep=huge,row sep=large]
x^{\ast}(E)
\arrow[rr,"\alpha"]
\arrow[dr,""'] && y^{\ast}(F)  
\arrow[dl,""] \\
&  \Delta^n&  
\end{tikzcd}
$$
We denote such a simplex by $(y; x, \alpha)$. To describe the simplicial structure on $E(f, g)$, consider a simplex $y \in Y_n$ and an ordinal map $\theta \colon [m] \to [n]$. These induce the following pullback square:
\begin{equation}\label{dia:thetainEfg}
\begin{tikzcd}[column sep=huge, row sep=large]
y^{\ast}(F) \arrow[d,"g_y"] & (y \circ \theta)^{\ast}(F) \arrow[l] \arrow[d,"g_{y \circ \theta}"] \\
\Delta^n & \Delta^m \arrow[l,"\theta"]
\end{tikzcd}
\end{equation}
Diagram~\eqref{dia:thetainEfg} corresponds to a morphism $(\theta, \Id) \colon g_y \to g_{y \circ \theta}$ in the category $\catsScen$. Consequently, we obtain maps
\[
\catsScen\left(f, (\theta, \Id)\right) = (\theta, \Id)_\ast \colon \catsScen(f, g_y) \to \catsScen(f, g_{y \circ \theta})
\]
for each $y \in Y_n$ and $\theta \colon [m] \to [n]$. These maps collectively define the simplicial structure on the set $E(f, g)$.

{\begin{rem}
The category $\catsScen$ contains the category of simplicial sets as a full subcategory: every simplicial set $X$ can be viewed as the simplicial bundle scenario $X \to \ast$. Moreover, for such objects we have a natural isomorphism
\[
\catsScen(X \to \ast,\, Y \to \ast) \cong \catsSet(X, Y),
\]
with composition preserved accordingly. In particular, we have:
\begin{equation}\label{eq:MapandMap}
\Map(X \to \ast,\, Y \to \ast) \cong \Map(X, Y) \to \ast,
\end{equation}
where the left-hand side denotes the {mapping scenario} between bundle scenarios, and the right-hand side is the {standard} mapping space of simplicial sets, e.g., see \cite{goerss2009simplicial}, regarded as a scenario over $\ast$.
\end{rem}
}

We define a symmetric monoidal structure on the category of simplicial bundle scenarios. The monoidal product of two scenarios 
\[
f \colon E \to X \quad \text{and} \quad g \colon F \to Y
\]
is given by
\[
f \otimes g \colon E \times F \xrightarrow{f \times g} X \times Y.
\]
The monoidal unit is the identity map $\Id_{\Delta^0}$ on the $0$-simplex. {Note that this symmetric monoidal structure is inherited from the one defined in Equation (\ref{eq:tensor stochastic}) on the category of stochastic simplicial scenarios, via the embedding given in (\ref{eq:deltasScensScenD}).}

{As the next result shows, however, our simplicial mapping scenario construction fails to be closed monoidal with respect to this monoidal structure.}

\begin{pro}\label{pro:tensor vs mapping simplicial}
There is a natural inclusion
\begin{equation}\label{eq:AdjunMap}
\catsScen(f\otimes g,h) \hookrightarrow \catsScen(f,\Map(g,h))   .
\end{equation}
\end{pro}
\begin{proof}
To explain the inclusion in (\ref{eq:AdjunMap}), consider the maps $f\colon E\to X$, $g\colon F \to Y$, and $h\colon G \to Z$. For 
$
({(\pi',\pi'')},\alpha) \in \catsScen(f\otimes g,h)
$
 as illustrated in the following diagram:
\begin{equation}\label{dia:ftimesgtoh}
\begin{tikzcd}[column sep=huge,row sep=large]
E\times F \arrow[d,"f \times g"'] &  E\times_X Z \times_{Y} F  
\arrow[l] \arrow[d,""] \arrow[r,"\alpha"] &  \arrow[ld,"h"] G \\
X \times Y & \arrow[l,"\text{${(\pi',\pi'')}$}"] Z 
\end{tikzcd}
\end{equation}
We construct the morphism $(\pi',\beta)\in \catsScen(f,\Map(g,h))$ as follows: For each level $n$, define:
$$
\beta_n \colon E_n\times_{X_n} Z_n \to E(g,h)_n
$$ 
which maps $(e,z)$ to $(z;\pi''(z)\colon \Delta^n \to Y, \beta')$, where 
$\beta_n'(\tilde{e})=\alpha_n(e,z,\tilde{e})$ for every $\tilde{e}\in g_n^{-1}\left(\pi''(z)\right)$.
%
\end{proof}

{Here is an example illustrating why the map in~(\ref{eq:AdjunMap}) is not a bijection in general. The example is formulated at the level of sets, which can be viewed as a special case of simplicial sets---specifically, as discrete simplicial sets.}

\begin{ex}
We define the sets $X = \{x\}$, $Y = \{y_1, y_2\}$, $Z = \{z\}$, $E = \{e_1, e_2\}$, $F = \{s, t_1, t_2\}$, and $G = \{u_1, u_2\}$. We also consider the maps
\[
f\colon E \to X, \quad g\colon F \to Y, \quad h\colon G \to Z,
\]
where $h(s) = y_1$ and $h(t_1) = h(t_2) = y_2$. For $(\pi, \alpha) \in \catsScen(f \otimes g, h)$, we have the following two cases:
\begin{itemize}
    \item If $\pi(z) = (x, y_1)$, then $\alpha$ can be any set map from $\{(e_1, s), (e_2, s)\}$ to $\{u_1, u_2\}$.
    \item If $\pi(z) = (x, y_2)$, then $\alpha$ can be any set map from $\{(e_1, t_1), (e_1, t_2), (e_2, t_1), (e_2, t_2)\}$ to $\{u_1, u_2\}$.
\end{itemize}
{So the total number of morphisms is
\[
|\catsScen(f \otimes g, h)| = 2^2 + 2^4 = 20.
\]
Now, we compute $|\catsScen(f, \Map(g, h))|$. Since $|Z| = 1$, we have 
\[
E(g, h) = \catsScen(g, h),
\]
which contains $6$ morphisms. Therefore, there is a one-to-one correspondence between $\catsScen(f, \Map(g, h))$ and $\catSet(\{e_1, e_2\}, \catsScen(g, h))$. As a result,
\[
|\catsScen(f, \Map(g, h))| = 6^2 = 36.
\]
}
\end{ex}

Nevertheless, 
we have the following important bijection that will be used later:
\begin{equation}\label{eq:Idtimesfg}
\catsScen(\Id_{\Delta^0}\otimes f,g) \cong \catsScen(\Id_{\Delta^0},\Map(f,g)).
\end{equation}

We now compare the mapping framework for simplicial bundle scenarios with the one introduced for bundle scenarios in Section~\ref{subsec:MapBund}. {To this end,} we introduce a nerve functor that assigns to each simplicial complex a corresponding simplicial set.


 \begin{defn}\label{def:N}
Let $N\colon \catsComp \to \catsSet$ be the functor that sends a simplicial complex $\Sigma$ to a simplicial set $N\Sigma$, called the \emph{nerve space} of $\Sigma$. The set of $n$-simplices of $N\Sigma$ is defined by
\[
(N\Sigma)_n = \left\{ (\sigma_1, \sigma_2, \ldots, \sigma_n) \in \Sigma^n \;\middle|\; \bigcup_{i=1}^n \sigma_i \in \Sigma \right\}, \quad n \geq 1,
\]
and we define ${V(N\Sigma)} = \{()\}$.

The simplicial structure maps are defined as follows:
\begin{itemize}
    \item {Face maps:}
    \[
    d_i(\sigma_1, \sigma_2, \ldots, \sigma_n) = 
    \begin{cases}
        (\sigma_2, \ldots, \sigma_n) & i = 0, \\
        (\sigma_1, \ldots, \sigma_i \cup \sigma_{i+1}, \ldots, \sigma_n) & 0 < i < n, \\
        (\sigma_1, \ldots, \sigma_{n-1}) & i = n.
    \end{cases}
    \]
    
    \item {Degeneracy maps:}
    \[
    s_j(\sigma_1, \sigma_2, \ldots, \sigma_n) = (\sigma_1, \ldots, \sigma_j, \emptyset, \sigma_{j+1}, \ldots, \sigma_n), \quad 0 \leq j \leq n.
    \]
\end{itemize}

Furthermore, a simplicial complex map $f\colon \Gamma \to \Sigma$ induces a simplicial set map $Nf\colon N\Gamma \to N\Sigma$, defined in degree $n$ by
\[
(Nf)_n(\gamma_1, \gamma_2, \ldots, \gamma_n) = \big(f(\gamma_1), f(\gamma_2), \ldots, f(\gamma_n)\big).
\]
\end{defn}

{The simplicial sets $N[f, g]$ and $\Map(Nf, Ng)$ are closely related but not, in general, isomorphic. In the next result, we make this relationship precise by constructing explicit maps between them.
}

\begin{pro}\label{pro:NMap=MapN}
{
There exists simplicial set maps 
\[
l\colon N{[f, g]} \to \Map(Nf, Ng) \quad \text{and} \quad t\colon \Map(Nf, Ng) \to N{[f, g]},
\]
such that $l \circ t = \Id_{\Map(Nf, Ng)}$.
}
\end{pro}

\begin{proof}
{We begin by defining the maps $l$ and $t$.} Given bundle scenarios $f \colon \Gamma \to \Sigma$ and $g \colon \Gamma' \to \Sigma'$, we define the map:
$$
l_n\colon  \left(N\Gamma(f,g) \right)_n \to E(Nf,Ng)_n
$$ 
by
$$
l_n\left((\sigma_1;\pi_1,\alpha^{(1)}), \cdots, 
(\sigma_n;\pi_n,\alpha^{(n)})\right)=\left((\sigma_1, \cdots, \sigma_n);
(\overline{\pi}_1(\sigma_1),\cdots,\overline{\pi}_n(\sigma_n)),\beta\right)
$$
where $\overline{\pi}_i(\sigma_i)=\cup_{\tau \in  \pi_i(\sigma_i)} \tau$ and $\beta_n$ is described as {follows}:
$$
\beta_n(\gamma_1,\cdots,\gamma_n)=\left(\alpha^{(1)}\left(L^{-1}_{\pi_1(\sigma_1)}(\gamma_1)\right),\cdots, 
\alpha^{(n)}\left(L^{-1}_{\pi_n(\sigma_n)}(\gamma_n)\right)\right),
$$
see Lemma \ref{lem:LhatNf}. 
{We} define 
$$
t_n\colon E(Nf,Ng)_n \to  \left(N\Gamma(f,g) \right)_n
$$ 
by
$$
t_n\left((\sigma_1, \cdots, \sigma_n);
(\tau_1,\cdots,\tau_n),\alpha\right)
=\left((\sigma_1;\set{\tau_1},\alpha_1|_{f^{-1}(\tau_1)}),\cdots,(\sigma_n;\set{\tau_n},\alpha_1|_{f^{-1}(\tau_n)})\right).
$$
The verification that these maps are simplicial is straightforward. We show that $l_n \circ t_n =\Id_{E(Nf,Ng)_n}$ for every $n \geq 0$. Given 
$\left((\sigma_1, \cdots, \sigma_n);
(\tau_1,\cdots,\tau_n),\alpha\right) \in E(Nf,Ng)_n$, we have
$$
\begin{aligned}
l_n\circ t_n\left((\sigma_1, \cdots, \sigma_n);
(\tau_1,\cdots,\tau_n),\alpha\right)
&=l_n\left((\sigma_1;\set{\tau_1},\alpha_1|_{f^{-1}(\tau_1)}),\cdots,(\sigma_n;\set{\tau_n},\alpha_n|_{f^{-1}(\tau_n)})\right)  \\
&=\left((\sigma_1, \cdots, \sigma_n);
(\tau_1,\cdots,\tau_n),\beta\right)
\end{aligned}
$$
where 
$$
\begin{aligned}
\beta_n(\gamma_1,\cdots,\gamma_n)
&=\left(\alpha_1\left(L^{-1}_{\set{\tau_1}}(\gamma_1)\right),\cdots, 
\alpha_1\left(L^{-1}_{\set{\tau_n}}(\gamma_n)\right)\right) \\
&=\left(\alpha_1(\gamma_1),\cdots, 
\alpha_1(\gamma_n)\right) \\
&=\alpha_n(\gamma_1,\cdots,\gamma_n) .
\end{aligned}
$$
\end{proof}

\subsection{{Convex maps between simplicial distributions}}
 

{
In this section, we characterize convex maps between simplicial distributions as convex combinations of maps induced by morphisms in the category of simplicial scenarios. This extends the main result of \cite[Theorem~44]{barbosa2023closing}, which addresses the case of standard measurement scenarios.
Our perspective highlights the structural role of the mapping space $\Map(f,g)$ in capturing the convex morphisms between simplicial distributions.
}

{We begin by observing that morphisms between simplicial scenarios induce convex maps between the corresponding convex sets of simplicial distributions.}
The bijection in (\ref{eq:Idtimesfg}) gives
\begin{align*}
\catsScen(f,g) &= \catsScen(\Id_{\Delta^0} \otimes f, g) \\
& \cong \catsScen(\Id_{\Delta^0}, \Map(f,g))\\
& = \sSect(\Map(f,g)).
\end{align*}
We denote this natural isomorphism by
\begin{equation}\label{eq:etafg}
\zeta_{f,g} \colon \catsScen(f,g) \xrightarrow{\cong} \sSect(\Map(f,g)).
\end{equation}
{On the other hand, the functor $\sDist \colon \catsScen \to \catConv$ introduced in Section~\ref{sec:SimpDistFunc} induces} 
\[
\sDist_{f,g} \colon \catsScen(f,g) \to \catConv(\sDist(f), \sDist(g)).
\]
Since $\catConv(\sDist(f), \sDist(g))$ is a convex set, we obtain the convex map
\begin{equation}\label{eq:transsDisfg}
\sDist_{f,g}' \colon D\left(\catsScen(f,g)\right) \to \catConv\left(\sDist(f), \sDist(g)\right),
\end{equation}
which is the transpose of $\sDist_{f,g}$ under the adjunction in (\ref{eq:Set adjunction Conv}).
{Combining} 
the maps in (\ref{eq:etafg}) and (\ref{eq:transsDisfg}), we conclude that a convex combination of sections on {the mapping scenario} $\Map(f,g)$ 
{induces} 
a convex map  $\sDist(f)\to\sDist(g)$. 

{Next,} we extend this {observation} to simplicial distributions on $\Map(f,g)$.
Given a simplicial bundle scenario $f \colon E \to X$, we define a diagram $\chi_f$ in 
{in the category of sets}
consisting of the following maps:
\[
d_i^E \colon f_n^{-1}(x) \to f_{n-1}^{-1}(d_i(x)) \quad \text{and} \quad s_j^E \colon f_{n-1}^{-1}(x) \to f_n^{-1}(s_j(x)),
\]
where {$n \geq 0$}, $x \in X_n$, and $d_i$, $s_j$ denote the face and degeneracy maps, respectively.

\begin{pro}\label{pro:limitchilimitDchi}
The limit of the diagram $\chi_f$ 
in $\catSet$ is isomorphic to $\sSect(f)$, and the limit of $D(\chi_{f})$ in $\catConv$ is isomorphic to $\sDist(f)$.   
\end{pro}
Given simplicial bundle scenarios $f\colon E\to X$ and $g\colon F\to Y$, for $y \in Y_n$ we have 
$$
\Map(f,g)_n^{-1}(y)=\catsScen(f,g_y)
$$
Moreover, we have the map 
$$
\sDist_{f,g_y}\colon \catsScen(f,g_y) \to \catConv(\sDist(f),\sDist(g_y))
$$
{where the target can be identified with $\catConv(\sDist(f),D(g_n^{-1}(y)))$ since $\sDist(g_y)$ can be identified with $D(g_n^{-1}(y))$.}
Due to the functoriality of $\sDist$, for every simplex $y \in Y_n$ and ordinal map $\theta \colon [m] \to [n]$, the following diagram commutes:
\begin{equation}\label{dia:Diagram1}
\begin{tikzcd}[column sep=huge,row sep=large]
\catsScen(f,g_y) \arrow[d,"\text{$(\theta,\Id)_\ast$}"'] \arrow[rr,"\sDist_{f,g_y}"] && \catConv\left(\sDist(f),D(g_n^{-1}(y))\right) 
\arrow[d,"\text{$\sDist(\theta,\Id)_\ast$}"] \\
\catsScen(f,g_{y\circ \theta})  \arrow[rr,"\sDist_{f,g_{y\circ \theta}}"] && \catConv\left(\sDist(f),D(g_m^{-1}(y\circ \theta))\right) 
\end{tikzcd}
\end{equation}
%
{This diagram} induces the following diagram: 
$$
\begin{tikzcd}[column sep=huge,row sep=large]
\catsScen(f,g_y) \arrow[dddr,hook,"\delta_{\catsScen(f,g_y)}"] \arrow[dd,"\text{$(\theta,\Id)_\ast$}"'] \arrow[rr,"\sDist_{f,g_y}"] && \catConv\left(\sDist(f),D(g^{-1}(y))\right) 
\arrow[dd,"\text{$\sDist(\theta,\Id)_\ast$}"] \\
&&  \\
\catsScen(f,g_{y\circ \theta})  \arrow[dddr,hook,"\delta_{\catsScen(f,g_{y\circ \theta})}"']  \arrow[rr,"\sDist_{f,g_{y\circ \theta}}"] && 
\catConv\left(\sDist(f),D(g^{-1}(y\circ \theta))\right)  \\
& D(\catsScen(f,g_y)) \arrow[dd,"\text{$D((\theta,\Id)_\ast)$}"'']  \arrow[ruuu,"\sDist_{f,g_y}'"]  \\ 
&& \\
& D(\catsScen(f,g_{y\circ \theta}))  \arrow[ruuu,"{\sDist_{f,g_{y\circ \theta}}'}"']
\end{tikzcd}
$$
%
%
{Therefore} we have the following 
{commutative} diagram: 
\begin{equation}\label{dia:Diagramofdiagrams}
\begin{tikzcd}[column sep=huge,row sep=large]
\chi_{\Map(f,g)} \arrow[dr,hook,"\delta_{\chi_{\Map(f,g)}}"'] \arrow[rr,""] && \catConv(\sDist(f),D(\chi_{g})) \\
& D(\chi_{\Map(f,g)})   \arrow[ru]
\end{tikzcd}
\end{equation}
%
%
By taking the limit 
{of}
Diagram (\ref{dia:Diagramofdiagrams}) and using Proposition \ref{pro:limitchilimitDchi}, we obtain the following commutative diagram:
%
\begin{equation}\label{dia:limitDiagramofdiagrams}
\begin{tikzcd}[column sep=huge,row sep=large]
\sSect(\Map(f,g)) \arrow[dr,hook,"(\delta_{E(f,g)})_\ast"'] \arrow[rr,""] && \catConv(\sDist(f),\sDist(g)) \\
& \sDist(\Map(f,g))  \arrow[ru,"\mu_{f,g}"']
\end{tikzcd}
\end{equation}
The top map in Diagram (\ref{dia:limitDiagramofdiagrams}) is equal to $\sDist_{f,g}\circ \zeta_{f,g}^{-1}$. 
{This way}
we have obtained a convex map 
$$
\mu_{f,g}\colon \sDist(\Map(f,g)) \to \catConv(\sDist(f),\sDist(g))
$$ 
which extends $\sDist_{f,g}$ in the following sense:
\begin{pro}    
The following diagram commutes:
\begin{equation}\label{dia:sDisFfg1}
\begin{tikzcd}[column sep=huge,row sep=large]
\catsScen(f,g) \arrow[rr,"\sDist_{f,g}"]  \arrow[drr,hook,"(\delta_{E(f,g)})_\ast\circ \,\zeta_{f,g}"']
&&  \catConv(\sDist(f),\sDist(g))  \\
&&\sDist(\Map(f,g)) \arrow[u,"\mu_{f,g}"']
\end{tikzcd}
\end{equation}
{Here, $\zeta_{f,g}$ is the isomorphism defined in (\ref{eq:etafg}).}

\end{pro}

Using Diagram (\ref{dia:sDisFfg1}), the naturality of $\delta$, and the definition of $\Theta$, we obtain the following commuting square:
\begin{equation}\label{dia:sDisFfg2}
\begin{tikzcd}[column sep=huge,row sep=large]
\catsScen(f,g)\arrow[d,hook,"\delta_{\catsScen(f,g)}"] \arrow[rr,"\sDist_{f,g}"]  
&&  \catConv(\sDist(f),\sDist(g))  \\
D(\catsScen(f,g))  \arrow[rr,"\Theta_{\Map(f,g)} \circ D(\zeta_{f,g})"] &&\sDist(\Map(f,g)) \arrow[u,"\mu_{f,g}"']
\end{tikzcd}
\end{equation}

Now, we present a generalization of \cite[Theorem 44]{barbosa2023closing}.

\begin{thm}\label{thm:contexMapp}
Let $f$ and $g$ be simplicial bundle scenarios. A convex map 
\[
\varphi \colon \sDist(f) \to \sDist(g)
\]
is a convex combination of maps induced by morphisms in $\catsScen(f,g)$ if and only if there exists a noncontextual simplicial distribution $p$ on $\Map(f,g)$ such that
\[
\varphi = \mu_{f,g}(p).
\]
\end{thm}
\begin{proof}
Since the composition $\mu_{f,g}\circ \Theta_{\Map(f,g)} \circ D(\zeta_{f,g})$ lies in $\catConv$, by Diagram 
(\ref{dia:sDisFfg2}), we get that 
\begin{equation}\label{eq:FThetasDist}
\mu_{f,g}\circ \Theta_{\Map(f,g)} \circ D(\zeta_{f,g})=\sDist'_{f,g}.
\end{equation}
Using Equation (\ref{eq:FThetasDist}) and the fact that $\zeta_{f,g}$ is bijective, we obtain the desired result.
\end{proof}

{For details on how this result extends \cite[Theorem 44]{barbosa2023closing} see Section \ref{sec:convex maps between empirical models}.}

\bibliography{bib.bib}
\bibliographystyle{ieeetr}

\normalfont 

\appendix

\section{Grothendieck construction}
\label{sec:gro}

In this section, we introduce the Grothendieck construction \cite{mac2013categories} and a relative variant that extends to $2$-categories \cite[Section XII.3]{mac2013categories}.

Let $\catC$ be a category. We will write {$\catCat$} for the category of small categories. Let $F\colon \catC\to \catCat$ be a functor. The Grothendieck construction of $F$ is the category $\int_{\catC} F$ consisting of:
\begin{itemize}
\item objects given by pairs $(c,x)$ where $c$ is an object of $\catC$ and $x$ is an object of $F(c)$,
\item morphisms $(c,x)\to (d,y)$ given by pairs $(h,\gamma)$ where $h\colon c\to d$ is a morphism of $\catC$ and $\gamma\colon F(h)(x) \to y$ is a morphism of $F(d)$.
\end{itemize}
For simplicity of notation we will sometimes omit the underlying category and write $\int F$. There is a canonical functor obtained by projecting onto the first factors
\[
\int F\to \catC.
\]
There is also a covariant version of the Grothendieck construction which applies to a functor of the form $F\colon \catC^\op \to \catCat$. The objects of $\int F$ in this case consists of pairs $(c,x)$ as before. A morphism $(c,x)\to (d,y)$ is given by a morphism {$h\colon d\to c$} of $\catC$ and a morphism {$\gamma\colon F(h)(x) \to y$} of $F(d)$.

A {crucial} fact we will use is that if $F,G\colon \catC\to \catCat$ are two functors together with a natural isomorphism $F\to G$ then there is an equivalence between the associated Grothendieck constructions
\[
\int F \xrightarrow{\simeq} \int G.
\] 

An important special case of the (contravariant) Grothendieck construction is when the functor $F$ takes values in sets instead of categories. The Grothendieck construction of such a functor $F\colon\catC\to \catSet$ is called the \emph{category of elements}. We will write $\catC_F$ for this category. Its objects consists of pairs $(c,x)$ where $c$ is an object of $\catC$ and $x\in F(c)$. A morphism $(c,x)\to (d,y)$ in this category is given by a morphism ${h}\colon c\to d$ of $\catC$ such that $F(h)(x)=y$.

\subsection{Relative Grothendieck construction}
\label{sec:grothendieck}

In this section, we introduce a relative version of the Grothendieck construction which uses additional data from $2$-categorical structures. A $2$-category consists of:
\begin{itemize}
\item objects,
\item $1$-morphisms between objects,
\item $2$-morphisms between morphisms. 
\end{itemize}
We will consider strict $2$-categories, also known as categories enriched over the category of small categories equipped with the Cartesian monoidal structure. A canonical example is $\catCat$, whose objects are small categories, whose 1-morphisms are functors, and whose 2-morphisms are natural transformations.

\begin{defn}\label{def:Sli2cat} 
Let $\catC$ be a $2$-category, and let $c$ be an object of $\catC$. 
The \emph{{thick} slice category}
$\catC\slice c$ of $\catC$ over $c$ is defined as follows:
\begin{itemize}
    \item The objects of $\catC{\slice}c$ are morphisms of $\catC$ of the form $f\colon d\to c$.
    \item A morphism from $f\colon d\to c$ to $g\colon e \to c$ is a pair $(h,\eta)$, where $h\colon d \to e$ 
    is a morphism of $\catC$, and $\alpha\colon f \to g\circ h$ is a $2$-morphism. See Diagram (\ref{dia:mapinSlice}).
\end{itemize}
\begin{equation}\label{dia:mapinSlice}
    \begin{tikzcd}[column sep=huge,row sep=large]
d
\arrow[rr,"h"]
\arrow[dr,"f"',""{name=A,right}] && e
\arrow[dl,"g",""{name=B,left}] 
 \arrow[Rightarrow, from=A, to=B, "\eta"]\\
&  c &  
\end{tikzcd}
\end{equation}

We define the functor $\Pi\colon\catC{\slice}c \to \catC$ that maps an 
object $f\colon d\to c$ to $c$, and a morphism $(h,\eta)$ to $h$. Now, we can define a relative version of the Grothendieck construction.
\end{defn}
\begin{defn}\label{def:GenGroth}
Given categories $\catC$ and $\catE$, and a functor 
$
F\colon \catC \to \catCat{\slice}\catE
$, let $\bar{F}$ denote the composition $\Pi \circ F\colon \catC \to \catCat$. 
We define the \emph{relative Grothendieck construction} to be the functor 
\[
\int_{\catC}F\colon \int_{\catC} \bar{F} \to \catE,
\] 
defined by:
\begin{itemize}
    \item For an object $(c,x)$ of $\int_{\catC} \bar{F}$, consisting of an object $c$ of $\catC$ and an object of $x$ of $\bar{F}(c)$, we set $(\int_{\catC} F)(c,x)= F(c)(x)$. 
    \item For a morphism $(h,\gamma)\colon (c,x) \to (d,y)$ of $\int_{\catC} \bar{F}$, consisting of a morphism $h\colon c \to d$ of $\catC$ and 
 a morphism   $\gamma\colon \bar{F}(h)(x) \to y$ of $\bar{F}(d)$. If $F(h)=(\bar{F}(h),\eta)$,  
then we set
    $$
    (\int_{\catC} F)(h,\gamma)= F(d)(\gamma) \circ \eta_x.
    $$
\end{itemize}
\end{defn}

Note that the last composite has the form
\[
F(c)(x) \xrightarrow{\eta_x} F(d)(\bar{F}(h)(x))
\xrightarrow{F(d)(\gamma)} F(d)(y).
\]
where
\begin{itemize}
\item $\eta$ is a natural transformation $F(c)\to F(d) \circ \bar{F}(h)$, and
\item $F(d)$ is a functor $\bar{F}(d)\to E$.
\end{itemize}


%

\section{Properties of the distribution monad}
\label{sec:properties of distribution monad}

\begin{proof}[{\bf Proof of Proposition \ref{pro:distribution monad gluing}}]
This section $m$ defined in \cite[Equation $(15)$]{kharoof2023homotopical} and the third {property} is proved as part $(1)$ of Lemma 3.26 in the same paper. We prove the remaining axioms. 
For {property} $(2)$, we prove a stronger property. Given the commutative diagram 
\begin{equation}\label{eq:NaturalDDD}
\begin{tikzcd}[column sep=huge,row sep=large] 
  & Y \arrow[d,"g"]
\arrow[rd,"\beta"]  & \\
X \arrow[r,"f"] \arrow[rd,"\alpha"'] & Z  \arrow[rd,hook,"\gamma"]  
&  Y' \arrow[d,"g{'}"] \\
&  X'  \arrow[r,"f{'}"']  & Z' 
\end{tikzcd}
\end{equation}
in $\catSet$, we prove that 
$$
T(\alpha\times \beta)\circ m_{f,g}=m_{f',g'} \circ (T(\alpha)\times T(\beta)).
$$
Indeed, for $p,q \in D(X)\times_{D(Z)} D(Y)$ and $(x',y') \in X'\times_{Z'} Y'$, we have the following:
\begin{equation}\label{eq:Naturalityy1}
\begin{aligned}
D(\alpha\times \beta) \left(m_{f,g}(p,q)\right)(x',y')&=
\sum_{(\alpha \times \beta)(x,y)=(x',y')} m_{f,g}(p,q)(x,y) \\
&=\sum_{(\alpha \times \beta)(x,y)=(x',y')} \frac{p(x)q(y)}{D(f)(p)(f(x))}.
\end{aligned}
\end{equation}
On the other hand, we have
\begin{equation}\label{eq:Naturalityy2}
\begin{aligned}
m_{f',g'} \left((D(\alpha)(p),D(\beta)(q)\right)(x',y')&=
\frac{D(\alpha)(p)(x')\,D(\beta)(q)(y')}{D(f')\left(D(\alpha)(p)\right)(f'(x'))} \\
&= \sum_{\alpha(x)=x'}\frac{p(x)}{D(\gamma)\left(D(f)(p)\right)(f'(x'))} \, \sum_{\beta(y)=y'}q(y) \\
&= \sum_{\alpha(x)=x'}\frac{p(x)}{\sum_{\gamma(z)=f'(x')}D(f)(p)(z)} \, \sum_{\beta(y)=y'}q(y)  \\
&= \sum_{\alpha(x)=x'}\frac{p(x)}{D(f)(p)(f(x))} \, \sum_{\beta(y)=y'}q(y) \\
&=
\sum_{\alpha(x)=x', \; \beta(y)=y'} \frac{p(x)q(y)}{D(f)(p)(f(x))}.
\end{aligned}
\end{equation}
In the fourth equation we used the injectivity of $\gamma$. Again, by the injectivity of $\gamma$, the equations $\alpha(x)=x'$,  $\beta(y)=y'$ implies that $(x,y) \in X\times_Z Y$. Therefore, the last parts of equations (\ref{eq:Naturalityy1}) and (\ref{eq:Naturalityy2}) are equal.  

Now, we prove {property} $(4)$. Given $(x,y), (x',y') \in X \times_{Z}Y$, we have 
$$
\begin{aligned}
m(\delta_X \times \delta_Y(x,y))(x',y')&=m(\delta^x,\delta^y)(x',y')
=\frac{\delta^x(x')\delta^y(y')}{D(f)(\delta^x)(f(x'))}\\
&= \frac{\delta^x(x')\delta^y(y')}{\delta^{f(x)}(f(x'))} =\begin{cases}
1 &    (x',y')=(x,y) \\
0 &     \;\text{otherwise.}
\end{cases}
\end{aligned}
$$
So $m(\delta_X \times \delta_Y(x,y))=\delta^{(x,y)}=
\delta_{X\times_Z Y}(x,y)$.

Finally, we will show that for $P \in D(D(X))$ and $q\in D(Y)$ where 
$D(D(f))(P)=D(D(g))(\delta^q)$, we get 
$$
m\circ(\mu_X \times \mu_Y)(P,\delta^{q}) =
\mu_{X\times_Z Y}\circ 
D(m) \circ m (P,\delta^{q}),
$$
{w}hich is stronger than {property} $(5)$. Given $(x,y) \in X\times_Z Y$, we have
\begin{equation}\label{eq:multD1}
\begin{aligned}
m\circ(\mu_X \times \mu_Y)(P,\delta^{q})(x,y)&=m(\mu_X(P),
\mu_Y(\delta^{q}))(x,y)\\
&=m(\mu_X(P),q)(x,y) \\
&=\frac{\mu_X(P)(x)q(y)}{D(g)(q)(g(y))}\\
&=\sum_{p'\in D(X)}P(p') 
\frac{p'(x)q(y)}{D(f)(g)(g(y))}.
\end{aligned}
\end{equation}
On the other hand,
\begin{equation}\label{eq:multD2}
\begin{aligned}
\mu_{X\times_Z Y}\left(D(m)(m(P,\delta^q)) \right)(x,y)&=
\sum_{s\in D(X\times_Z Y)} 
D(m)\left(m(P,\delta^q)\right)(s)s(x,y) \\
&=\sum_{s\in D(X\times_Z Y)} 
\sum_{m(p',q')=s}m(P,\delta^q)(p',q')s(x,y) \\
&=\sum_{(p',q')\in D(X)\times_{D(Z)}D(Y)} 
m(P,\delta^q)(p',q')m(p',q')(x,y) \\
&=\sum_{p'\in D(X) \;\text{s.t}\; 
D(f)(p')=D(g)(q)}m(P,\delta^q)(p',q) m(p',q)(x,y)
\\
&=\sum_{p'\in D(X) \; \text{s.t} 
\;D(f)(p')=D(g)(q)}P(p') 
\frac{p'(x)q(y)}{D(f)(g)(g(y))}.
\end{aligned}
\end{equation}
In the last equation in (\ref{eq:multD2}) we used the following argument 
$$
m(P,\delta^q)(p',q)=\frac{P(p')\delta^q(q)}{D(D(g))(D(g)(q))}=\frac{P(p')}{\delta^{D(g)(q)}(D(g)(q))}=P(p').
$$

Note that 
$$
\begin{aligned}
\sum_{p'\in D(X) \; \text{s.t} 
\;D(f)(p')=D(g)(q)}P(p')&=
D(D(f))(P)(D(g)(q))\\
&=D(D(g))(\delta^q)(D(g)(q))\\
&=\delta^{D(g)(q)}(\delta^q)(D(g)(q))=1.
\end{aligned}
$$
Therefore, the last parts of equations (\ref{eq:multD1}) and (\ref{eq:multD2}) are equal.

\end{proof}

\begin{proof}[{\bf Proof of Proposition \ref{pro:distribution monad deterministic}}]
Given Diagram (\ref{eq:NaturalDDD}) in $\catSet$ without requiring $\gamma$ to be injective, we prove that the following diagram commutes:
$$
\begin{tikzcd}[column sep=huge,row sep=large]
D(X)\times_{D(Z)} Y \arrow[d,"D(\alpha)\times \beta"] \arrow[r,"\eta_{f,g}"]&   D(X \times_Z Y)  
\arrow[d,"D(\alpha\times \beta)"]
 \\
D(X')\times_{D(Z')} Y'  \arrow[r,"\eta_{f',g'}"]&  D(X' \times_{Z'} Y')   
\end{tikzcd}
$$
Given $(p,\tilde{y}) \in D(X)\times_{D(Z)} Y$ and  $(x',y')\in X\times_Z Y$, we have the following:
$$
\begin{aligned}
D\left(\alpha \times \beta)(\eta(p,\tilde{y})\right)(x',y')&= \sum_{(\alpha \times \beta)(x,y)=(x',y')}\eta(p,\tilde{y})(x,y) \\
& =\sum_{(\alpha \times \beta)(x,y)=(x',y')}p(x)\cdot\delta^{\tilde{y}}(y) \\
&=\begin{cases}
\sum_{f(x)=g(\tilde{y}),\,
\alpha(x)=x'}p(x)   & y'= \beta(\tilde{y}) \\
0 &  y'\neq \beta(\tilde{y}).
\end{cases}
\end{aligned}
$$
On the other hand, we have
$$
\begin{aligned}
\eta\left(D(\alpha)(p),  \beta(\tilde{y})\right)(x',y')&= D(\alpha)(p)(x') \cdot\delta^{\beta(\tilde{y})}(y') \\
&=\begin{cases}
\sum_{\alpha(x)=x'}p(x)   & y'= \beta(\tilde{y}) \\
0 &  y'\neq \beta(\tilde{y}).
\end{cases}
\end{aligned}
$$
Since $(p,\tilde{y}) \in D(X)\times_{D(Z)} Y$, we have $D(f)(p)=D(g)(\delta^{\tilde{y}})=\delta^{g(\tilde{y})}$. That means that the support of $p$ is $\set{x:\,f(x)=g(\tilde{y})}$. We conclude that
$$
D\left(\alpha \times \beta)(\eta(p,\tilde{y})\right)(x',y')=\eta\left(D(\alpha)(p),  \beta(\tilde{y})\right)(x',y').
$$
\end{proof}

\section{Simplicial complexes}

\label{sec:simplicial complexes}

\subsection{{Simplicial relations}}

{The nerve complex functor
\[
\hat N \colon \catsComp \to \catsComp
\]
defined in Definition~\ref{def:hatN} forms a monad.} The unit $\delta_\Sigma \colon \Sigma \to \hat N\Sigma$ of the monad is given by sending each vertex $x$ to the singleton simplex it generates, i.e., $x \mapsto \{x\}$. 
{The} multiplication $\mu_\Sigma \colon \hat N^2\Sigma \to \hat N\Sigma$ is defined by mapping each vertex $\{\sigma_1, \cdots, \sigma_n\}$ to the union $\bigcup_{i=1}^n \sigma_i$.
We denote the associated Kleisli category by $\catsComp_{\hat N}$ and write $\diamond$ for Kleisli composition. A morphism $\pi \colon \Sigma' \to \hat N\Sigma$ in this category corresponds to a simplicial relation from $\Sigma'$ to $\Sigma$. Accordingly, we identify the category $\catsRel$ of simplicial complexes with simplicial relations as the Kleisli category $\catsComp_{\hat N}$.

On the functor point of view, we will denote the functor induced by such a simplicial complex map by 
\[\overline{\pi}\colon \catC_{\Sigma'} \to \catC_{\Sigma}\]
{where a simplex} $\sigma$ {is sent} to $\cup_{\tau \in \pi(\sigma)} \tau$. 
We have 
\begin{equation}\label{pianddoverlinepi}
\overline{\pi}(\delta_{\Sigma})=\Id_{\catC_{\Sigma}}, \;\;\; \overline{\pi_2 \diamond \pi_1}= \overline{\pi_2} \circ \overline{\pi_1}.
\end{equation}

\subsection{Pullback of bundles} 
\label{sec:pullback of bundles} 

Given a bundle scenario $f \colon \Gamma \to \Sigma$ and a simplicial complex map $\pi \colon \Sigma' \to \hat N \Sigma$, we define the bundle scenario $\pi^\ast(f)$ as the right vertical map in the following pullback square:
\begin{equation}\label{dia:piastf}
\begin{tikzcd}[column sep=huge,row sep=large]
\hat N\Gamma \arrow[d,"\hat N f"] &\arrow[l,""]  \pi^{\ast} (\hat N\Gamma)
\arrow[d,"\pi^\ast(f)"] 
 \\
\hat N\Sigma \arrow[ru, phantom, "\llcorner", very near end]  &\arrow[l,"\pi"]  \Sigma'
\end{tikzcd}
\end{equation}
see \cite[Propositions A.6 and A.7]{barbosa2023bundle}. In addition, a morphism $\alpha \colon f \to g$ between bundle scenarios induces a morphism $\pi^{\ast}(\alpha) \colon \pi^\ast(f) \to \pi^\ast(g)$ between the corresponding pullbacks: 
$$
\begin{tikzcd}[column sep=huge,row sep=large]
\pi^{\ast} (\hat N\Gamma)
\arrow[rr,"\pi^{\ast}(\alpha)"]
\arrow[dr,"\pi^\ast(f)"'] && \pi^{\ast} (\hat N\Gamma')
\arrow[dl,"\pi^{\ast}(g)"] \\
&  \Sigma'&  
\end{tikzcd}
$$
Moreover, it is straightforward to check that
a simplicial complex map $\pi\colon \Sigma' \to \hat N \Sigma$ induces a functor 
\[\pi^{\ast}\colon \Bund(\Sigma) \to  \Bund(\Sigma').\]    

\begin{lem}\label{lem:lempi1pi2}
{Let $\pi_1 \colon \Sigma' \to \Sigma$ and $\pi_2 \colon \Sigma'' \to \Sigma'$ be simplicial complex maps. Then the following properties hold:}
\begin{enumerate}
    \item The pullback $(\delta_{\Sigma})^{\ast} \colon \Bund(\Sigma) \to \Bund(\Sigma)$ along the unit map is the identity functor.
    \item Pullback is functorial with respect to Kleisli composition: 
    \[
    (\pi_1 \diamond \pi_2)^{\ast} = \pi_2^{\ast} \circ \pi_1^{\ast}.
    \]
\end{enumerate}
\end{lem}  
\begin{proof}
Part (1): For a bundle scenario $f\colon \Gamma \to \Sigma$, by \cite[Lemma A.9]{barbosa2023bundle} we have $(\delta_{\Sigma})^{\ast}(f)=f$. On the other hand, for a morphism $\alpha\colon f \to g$ in 
$\Bund(\Sigma)$, due to the naturality of $\delta$, we obtain the following commuting diagram: 
$$
\begin{tikzcd}[column sep=huge,row sep=large]
\hat N\Gamma
\arrow[rr,"\hat N \alpha"]
\arrow[dr,"\hat N f"']  && \hat N\Gamma'
\arrow[dl,""] \\
&  \hat N \Sigma 
&& \Gamma
\arrow[rr,"\alpha"]
\arrow[dr,"f"'] \arrow[lllu,hook',""]&& \Gamma'
\arrow[dl,"g"]  \arrow[lllu,hook',"\delta_{\Gamma'}"']\\
&& &&  \Sigma  \arrow[lllu,hook',"\delta_{\Sigma}"]&  
\end{tikzcd}
$$
Thus, $(\delta_{\Sigma})^{\ast}(\alpha)=\alpha$.

Part (2): For objects, using \cite[Lemma A.8]{barbosa2023bundle}, we obtain 
$$
(\pi_1 \diamond \pi_2)^{\ast}(f)=(\pi_2)^{\ast} \circ (\pi_1)^{\ast}(f)
$$
For morphisms, the equality holds by \cite[Equation (40)]{barbosa2023bundle}.
\end{proof}

\subsection{Monoidal structure on simplicial complexes}\label{subsec:MonComp}

{In this section, we define a product on the category of simplicial complexes that equips it with the structure of a cartesian monoidal category. This construction satisfies the property that the functor $\hat N$ (as defined in Section~\ref{sec:standard}) becomes a monoidal monad with respect to the resulting monoidal structure. As a result, we show that the induced Kleisli category inherits a symmetric monoidal structure.
}

\begin{defn}
For simplicial complexes $\Sigma_1$ and $\Sigma_2$, we define their product $\Sigma_1 \otimes \Sigma_2$ as the simplicial complex with:
\begin{itemize}
    \item $V(\Sigma_1 \otimes \Sigma_2) = V(\Sigma_1) \times V(\Sigma_2)$, the Cartesian product of the vertex sets.
    \item A subset $\sigma \subset V(\Sigma_1) \times V(\Sigma_2)$ is a simplex if and only if $\pr_1(\sigma) \in \Sigma_1$ and $\pr_2(\sigma) \in \Sigma_2$, where $\pr_1$ and $\pr_2$ denote the projections onto the first and second factors, respectively.
\end{itemize}
For a pair of simplicial complex maps $f \colon \Sigma_1' \to \Sigma_1$ and $g \colon \Sigma_2' \to \Sigma_2$, we define $f \otimes g$ to act as $f \times g$ on the vertices of {$\Sigma_1' \otimes \Sigma_2'$}.
\end{defn}

Note that if $\Sigma$ is generated by $\set{\sigma_1,\cdots,\sigma_n}$, and $\Sigma'$ is generated by 
$\set{\tau_1,\cdots,\tau_m}$, then $\Sigma \otimes \Sigma'$ is generated by 
$\set{\sigma_i \times \tau_j: 1\leq i \leq n,\; 1\leq j \leq m}$.

\begin{pro}
$(\catsComp,\otimes,\Delta^0)$ is a cartesian monoidal category.    
\end{pro} 

{For every pair of simplicial complexes $\Sigma_1$ and $\Sigma_2$, we define the map
\begin{equation}\label{eq:phiSigma}
\phi_{\Sigma_1, \Sigma_2} \colon \hat N \Sigma_1 \otimes \hat N \Sigma_2 \to \hat N (\Sigma_1 \otimes \Sigma_2)
\end{equation}
by sending each vertex $(\sigma, \tau)$ to the product $\sigma \times \tau$.
We now show that the map in Equation (\ref{eq:phiSigma}) is a well-defined simplicial complex map. First, note that $\sigma \times \tau$ is indeed a vertex of $\hat N(\Sigma_1 \otimes \Sigma_2)$ because $\pr_1(\sigma \times \tau) = \sigma$ and $\pr_2(\sigma \times \tau) = \tau$. 
Next, let $\tilde\sigma = \{(\sigma_1, \tau_1), \cdots, (\sigma_n, \tau_n)\}$ be a simplex in $\hat N \Sigma_1 \otimes \hat N \Sigma_2$. By definition, we have that $\pr_{1}(\tilde\sigma) \in \hat N \Sigma_1$ and $\pr_{2}(\tilde\sigma) \in \hat N \Sigma_2$, which implies that $\bigcup_{i=1}^n \sigma_i \in \Sigma_1$ and $\bigcup_{i=1}^n \tau_i \in \Sigma_2$.
By the definition of $\phi_{\Sigma_1, \Sigma_2}$, we have
\[
\phi_{\Sigma_1, \Sigma_2}(\tilde\sigma) = \{ \sigma_1 \times \tau_1, \cdots, \sigma_n \times \tau_n \}.
\]
To show that $\bigcup_{i=1}^n (\sigma_i \times \tau_i) \in \Sigma_1 \otimes \Sigma_2$, observe that
\[
\pr_1\left( \bigcup_{i=1}^n (\sigma_i \times \tau_i) \right) = \bigcup_{i=1}^n \sigma_i \in \Sigma_1, \quad
\pr_2\left( \bigcup_{i=1}^n (\sigma_i \times \tau_i) \right) = \bigcup_{i=1}^n \tau_i \in \Sigma_2.
\]
This confirms that $\phi_{\Sigma_1, \Sigma_2}$ is well-defined as a simplicial complex map.
}

\begin{pro}
The {monad} $\hat N\colon \catsComp \to \catsComp$, equipped with the natural map in (\ref{eq:phiSigma}) forms a monoidal monad on 
$(\catsComp,\otimes,\Delta^0)$.   
\end{pro}
\begin{proof}
We define the coherence maps to be $\delta_{\Delta^0}\colon \Delta^0 \to \hat N (\Delta^0)$ and $\phi_{\Sigma,\Sigma'}$ 
{using Equation (\ref{eq:phiSigma}).}
It is straightforward to verify the commutativity of the following coherence diagrams: 
$$
\begin{tikzcd}[column sep=huge,row sep=large]
 (\hat N\Sigma \otimes \hat N \Sigma') \otimes \hat N \Sigma'' \arrow[d,"\phi_{\Sigma,\Sigma'}\otimes \Id_{\Sigma''}"]
\arrow[r,""] & \hat N\Sigma \otimes (\hat N \Sigma' \otimes \hat N \Sigma'') \arrow[d,"\Id_{\Sigma} \otimes \phi_{\Sigma',\Sigma''}"]\\
\hat N(\Sigma \otimes \Sigma') \otimes \hat N \Sigma''  \arrow[d,"\phi_{\Sigma\otimes \Sigma',\Sigma''}"]
 & \hat N\Sigma \otimes \hat N (\Sigma' \otimes \Sigma'')   \arrow[d,"\phi_{\Sigma,\Sigma' \otimes \Sigma''}"]\\
\hat N\left((\Sigma \otimes \Sigma') \otimes \Sigma''\right) \arrow[r,""] & \hat N\left(\Sigma \otimes (\Sigma' \otimes \Sigma'')\right)
\end{tikzcd}%
$$
$$
\begin{tikzcd}[column sep=huge,row sep=large]
\Delta^0  \otimes \hat N \Sigma\arrow[d,""]    
\arrow[r,"\delta_{\Delta^0} \otimes \Id_{\hat N \Sigma} "]&   \hat N \Delta^0 \otimes \hat N \Sigma  
\arrow[d,"\phi_{\Delta^0,\Sigma}"]     &
\hat N \Sigma \otimes \Delta^0 \arrow[d,""]    
\arrow[r,"\Id_{\hat N \Sigma}\otimes \delta_{\Delta^0}"]&   \hat N \Sigma \otimes \hat N \Delta^0  
\arrow[d,"\phi_{\Sigma,\Delta^0}"]
 \\
\hat N \Sigma & \arrow[l,""]  \hat N ( \Delta^0 \otimes \Sigma )     &
\hat N \Sigma & \arrow[l,""]  \hat N (\Sigma \otimes \Delta^0) 
\end{tikzcd}
$$
Since these coherence conditions hold, $\hat N\colon (\catsComp,\otimes,\Delta^0) \to (\catsComp,\otimes,\Delta^0)$ is a lax monoidal functor. Finally, we verify that the unit $\delta$ and the multiplication $\mu$ are monoidal natural transformations. For the unit, we need to prove that the following diagrams commute: 
\begin{equation}\label{eq:deltamon}
\begin{tikzcd}[column sep=huge,row sep=large]
\Sigma \otimes \Sigma' \arrow[d,equal]    
\arrow[r,"\delta_{\Sigma} \otimes \delta_{\Sigma'}"] &   \hat N \Sigma \otimes \hat N \Sigma'  
\arrow[d,"\phi_{\Sigma,\Sigma'}"]     
&&
\Delta^0 \arrow[ld,equal] \arrow[rd,"\delta_{\Delta^0}"]    
 \\
\Sigma \otimes \Sigma'  \arrow[r,"\delta_{\Sigma \otimes \Sigma'}"]  & \hat N (\Sigma \otimes \Sigma')    
&
\Delta^0  \arrow[rr,"\delta_{\Delta^0}"] && \hat N \Delta^0 
\end{tikzcd}
\end{equation}
The left-hand diagram in (\ref{eq:deltamon}) commutes because for vertices $x \in V(\Sigma)$ and $y \in V(\Sigma')$, we have $\{(x,y)\}=\{x\}\times\{y\}$, and the right-hand diagram obviously commutes.
 For the multiplication, we need to prove that the following diagrams commute: 
\begin{equation}\label{eq:mumon}
\begin{tikzcd}[column sep=huge,row sep=large]
\hat N^2 \Sigma \otimes \hat N^2 \Sigma' \arrow[d,"\phi_{\hat N \Sigma,\hat N \Sigma'}"']    
\arrow[r,"\mu_{\Sigma} \otimes \mu_{\Sigma'}"] &   \hat N \Sigma \otimes \hat N \Sigma'  
\arrow[dd,"\phi_{\Sigma,\Sigma'}"]     
&
\Delta^0 \arrow[d,"\delta_{\Delta^0}"'] \arrow[rrdd,"\delta_{\Delta^0}"]    
 &\\
\hat N(\hat N \Sigma \otimes \hat N\Sigma')  \arrow[d,"\hat N(\phi_{\Sigma,\Sigma'})"'] &    
&
  \hat N\Delta^0 \arrow[d,"\delta_{\hat N \Delta^0}"'] && \\
\hat N^2(\Sigma \otimes \Sigma')  \arrow[r,"\mu_{\Sigma \otimes \Sigma'}"] & \hat N (\Sigma \otimes \Sigma') &
\hat N^2 \Delta^0  \arrow[rr,"\mu_{\Delta^0}"] && \hat N \Delta^0
\end{tikzcd}
\end{equation}
For the left-hand diagram in (\ref{eq:mumon}), given a vertex 
$(\set{\sigma_1,\cdots,\sigma_n},\set{\tau_1,\cdots,\tau_m})$ in $\hat N ^2 \Sigma \otimes \hat N^2 \Sigma'$, we have the following:
$$
\begin{aligned}
&\phi_{\Sigma,\Sigma'}\circ(\mu_{\Sigma}\otimes \mu_{\Sigma'})\left(\set{\sigma_1,\cdots,\sigma_n},\set{\tau_1,\cdots,\tau_m}\right)\\
&=\phi_{\Sigma,\Sigma'}(\cup_{i=1}^n \sigma_i,\cup_{i=1}^m\tau_i)\\
&=\cup_{i=1}^n \sigma_i\times \cup_{j=1}^m\tau_j\\
&=\cup_{i,j} \sigma_i \times \tau_j \\
&=\mu_{\Sigma \otimes \Sigma'}(\set{\sigma_i \times \tau_j:~i,j}) \\
&=\mu_{\Sigma \otimes \Sigma'}\circ \hat N(\phi_{\Sigma,\Sigma'})(\set{(\sigma_i,\tau_j):~i,j}) \\
&=\mu_{\Sigma \otimes \Sigma'}\circ \hat N(\phi_{\Sigma,\Sigma'})\left(\set{\sigma_1,\cdots,\sigma_n} \times \set{\tau_1,\cdots,\tau_m}\right)\\
&=\mu_{\Sigma \otimes \Sigma'}\circ \hat N(\phi_{\Sigma,\Sigma'}) \circ \phi_{\hat N \Sigma,\hat N \Sigma'}
\left(\set{\sigma_1,\cdots,\sigma_n},\set{\tau_1,\cdots,\tau_m}\right).
\end{aligned}
$$
The right-hand diagram in (\ref{eq:mumon}) commutes since 
$\mu_{\Delta^0}\circ \delta_{\hat N \Delta^0}=\Id_{\hat N \Delta^0}$.
\end{proof}

\begin{cor}\label{cor:KlisMonoi}
The Kleisli category $\catsComp_{\hat N} = \catsRel$ inherits a symmetric monoidal structure, with the tensor product defined as follows:
\begin{itemize}
    \item On objects, the tensor product is given by $\Sigma_1 \boxtimes \Sigma_2 = \Sigma_1 \otimes \Sigma_2$.
    
    \item On morphisms, for $\pi_1 \colon \Sigma_1 \to \hat N \Sigma'_1$ and $\pi_2 \colon \Sigma_2 \to \hat N \Sigma'_2$, their tensor product is defined by
    \begin{equation}\label{eq:pi1otimespi2}
    \pi_1 \boxtimes \pi_2 = \phi_{\Sigma'_1, \Sigma'_2} \circ (\pi_1 \otimes \pi_2).
    \end{equation}
\end{itemize}
\end{cor}

\section{{Comparison results}}
\label{sec:comparison results}

\subsection{{Comparison of mapping scenarios}}
\label{sec:comparison of mapping scenarios}

In \cite[Section~4]{barbosa2023closing}, for every pair of standard scenarios $S_1$ and $S_2$, the authors define a new scenario $[S_1, S_2]$ as follows:
\begin{itemize}
    \item $\Sigma_{[S_1, S_2]} = \Sigma_{S_2}$.
    \item $O_{[S_1, S_2], x} = \set{(U, \alpha) \mid U \subset V(\Sigma_{S_1}) \; \text{and} \; \alpha \colon \prod_{x \in U} O_{S_1,x} \to O_{S_2,x}}$ for every $x \in V(\Sigma_{[S_1, S_2]})$.
\end{itemize} 
An empirical model $p$ on $[S_1, S_2]$ satisfying the predicate $g_{S_1, S_2}$ (see \cite[Definition~4.1]{barbosa2023bundle}) consists of compatible distributions $\set{p_{\sigma}}_{\sigma \in \Sigma_{S_2}}$, where for 
$\sigma = \set{x_1, \cdots, x_n}$, $p_{\sigma}$ is a distribution on the following set:
\begin{equation}\label{eq:setFGgST}
\left\{ \set{(\tau_1, \alpha_1), \cdots, (\tau_n, \alpha_n)} \;\middle|\; \cup_{i=1}^n \tau_i \in \Sigma_{S_1} \; \text{and} \; \alpha_i \colon \prod_{y \in \tau_i} O_{S_1, y} \to O_{S_2, x_i} \right\}.
\end{equation}
We denote the set of empirical models on $[S_1, S_2]$ satisfying $g_{S_1, S_2}$ by $\Emp([S_1, S_2], g_{S_1, S_2})$. 
In fact, we can define a functor
\[
([S_1, S_2], g_{S_1, S_2}) \colon \catC^{\op}_{\Sigma_{S_2}} \to \catSet
\]
by assigning to each $\sigma \in \Sigma_{S_2}$ the set in~\eqref{eq:setFGgST}. This leads to the following result.

\begin{pro}
There is a natural isomorphism of functors
\[
([S_1, S_2], g_{S_1, S_2}) \cong [\eE_{S_1}, \eE_{S_2}].
\]
\end{pro}

\begin{proof}
By Remark \ref{rem:NatDeter}, if $\sigma=\set{x_1,\cdots,x_n}$, then an element $(\pi,\alpha) \in [\eE_S,\eE_T](\sigma)$ is determined by the simplex 
$\pi(\sigma)=\set{\pi(x_1),\cdots,\pi(x_n)} \in \hat N \Sigma_{T}$ and the set map $\alpha_{\sigma}$. Because of the naturality of 
$\alpha$, for every $1\leq j \leq n$ we have the following commuting diagram:
$$
\begin{tikzcd}[column sep=huge,row sep=large]
 \prod_{y\in \cup_{i=1}^n\pi(x_i)}O_{S,y} \arrow[d,""] \arrow[r,"\alpha_{\sigma}"]&    \prod_{i=1}^n O_{T,x_i}  
\arrow[d,""]
 \\
 \prod_{y\in \pi(x_j)}O_{S,y} \arrow[r,"\alpha_{x_j}"]&  O_{T,x_j}   
\end{tikzcd}
$$
This implies that $\alpha_{\sigma}$ is uniquely determined by the maps $\alpha_{x_j}$. Therefore, we can write 
$$
[\eE_S,\eE_T](\sigma)=\set{(\set{\tau_1,\cdots,\tau_n},\set{\alpha_{x_1},\cdots,\alpha_{x_n}}):\; \cup_{i=1}^n \tau_i \in \Sigma_S\; \text{and}\;\alpha_{x_j} \colon \prod_{y\in \tau_j}O_{S,y} \to O_{T,x_j}}
$$
which is naturally isomorphic to $([S,T],g_{S,T})(\sigma)$.  
\end{proof}
\begin{cor}
The functor $([S,T],g_{S,T})$ is {an event scenario}.   
\end{cor}
\begin{cor}\label{cor:EmpSTEmpeE}
There is a convex isomorphism between $\Emp([S_1, S_2], g_{S_1, S_2})$ and $\Emp[\eE_{S_1}, \eE_{S_2}]$, which preserves the (non)contextuality of empirical models (see Definition~\ref{def:conEmp}).
\end{cor}

\subsection{{Convex maps between empirical models}}
\label{sec:convex maps between empirical models}

We now explain how to obtain \cite[Theorem 44]{barbosa2023closing} using Theorem \ref{thm:contexMapp}. According to \cite[Theorems 3.18 and 4.18]{barbosa2023bundle}, 
there are fully faithful functors $\El(\eE)\colon \catScen \to \catbScen$ and 
$N\colon \catbScen \to \catsScen$, as well as natural isomorphisms $\eta$ and $\zeta$ as follows:
\begin{equation}\label{dia:BundTheorems}
    \begin{tikzcd}[column sep=huge,row sep=large]
\catScen \arrow[r,"\El(\eE)"]
\arrow[dr,"\Emp"',""{name=C,right}]
 & \catbScen
\arrow[r,"N"]
\arrow[d,"",""{name=A,right}] & \catsScen
\arrow[dl,"\sDist",""{name=B,left}] 
 \arrow[Rightarrow, from=C, to=A, "\eta"] \arrow[Rightarrow, from=A, to=B, "\zeta"]\\
&  \catConv &  
\end{tikzcd}
\end{equation}
where the middle vertical arrow is the functor $\bEmp$. As explained in \cite[Corollary 5.4]{barbosa2023bundle}, Diagram (\ref{dia:BundTheorems}) induces the following diagram: 

\begin{equation}\label{dia:BundCorollaries}
    \begin{tikzcd}[column sep=huge,row sep=large]
D(\catScen) \arrow[r,"D(\El(\eE))"]
\arrow[dr,"{\Emp'}"',""{name=C,right}]
 & D(\catbScen)
\arrow[r,"D(N)"]
\arrow[d,"",""{name=A,right}] & D(\catsScen)
\arrow[dl,"{\sDist'}",""{name=B,left}] 
 \arrow[Rightarrow, from=C, to=A, "{\eta'}"] \arrow[Rightarrow, from=A, to=B, "{\zeta'}"]\\
&  \catConv &  
\end{tikzcd}
\end{equation}
where $\Emp'$ and $\sDist'$ are the transposes of the functors of $\Emp$ and $\sDist$ under the adjunction $D:\catCat \adjoint \catConvCat :U$, respectively. Meanwhile, $\eta$ and $\zeta$
are natural isomorphisms.
This leads to the following commutative diagram:
\begin{equation}\label{dia:sDistEmp}
\begin{tikzcd}[column sep=huge,row sep=large]
D\left(\catScen(S,T)\right)  \arrow[d,"{\Emp'_{S,T}}"] \arrow[rr,"D(N\circ \El(\eE))_{S,T}","\cong"']  
&&  D\left(\catsScen(N\El(\eE_S),N\El(\eE_T))\right)  \arrow[d,"{\sDist'_{N\El(\eE_S),N\El(\eE_T)}}"] \\
\catConv\left(\Emp(S),\Emp(T)\right) \arrow[rr,"{(\zeta' \circ \eta')_T \circ -\circ (\zeta' \circ \eta')_S^{-1}}","\cong"'] && 
\catConv\left(\sDist(N\El(\eE_S)),\sDist(N\El(\eE_T))\right) 
\end{tikzcd}
\end{equation}

\begin{pro}\label{pro:sDistNfgMap}
There exists a surjective convex map
\[
\sDist(N{[f,g]}) \to \sDist(\Map(Nf,Ng))
\]
such that a simplicial distribution $p \in \sDist(\Map(Nf,Ng))$ is noncontextual if and only if it is the image of a noncontextual simplicial distribution in $\sDist(N{[f,g]})$.
\end{pro}

\begin{proof}
According to Proposition \ref{pro:NMap=MapN} the map $l\colon N{[f,g]} \to \Map(Nf,Ng)$ has a section, so it induces the following commutative diagram 
$$
\begin{tikzcd}[column sep=huge,row sep=large]
D(\sSect(N{[f,g]})) \arrow[d,"D(l_\ast)",tail, twoheadrightarrow] \arrow[rr,"\Theta_{N{[f,g]}}"]  
&&  \sDist(N{[f,g]}) \arrow[d,"l_\ast",tail, twoheadrightarrow] \\
D(\sSect(\Map(Nf,Ng))) \arrow[rr,"\Theta_{\Map(Nf,Ng)}"]   && \sDist(\Map(Nf,Ng)) 
\end{tikzcd}
$$
where the vertical map{s} are surjective.  {The desired result follows by diagram chasing.}
\end{proof}

\begin{pro}\label{pro:themapprescontex}
There exists a surjective convex map
\[
r_{S,T}\colon \Emp([S,T],g_{S,T}) \twoheadrightarrow \sDist(\Map(N\El(\eE_S),N\El(\eE_T)))
\]
such that a simplicial distribution $p \in \sDist(\Map(N\El(\eE_S),N\El(\eE_T)))$ is noncontextual if and only if there exists a noncontextual empirical model $q \in \Emp([S,T],g_{S,T})$ {satisfying} $r_{S,T}(q) = p$.
\end{pro}
%
\begin{proof}
By Corollary~\ref{cor:EmpSTEmpeE}, there is an isomorphism
\[
\Emp([S,T],g_{S,T}) \cong \Emp[\eE_S,\eE_T].
\]
Next, \cite[Theorem~5.10(1)]{barbosa2023bundle} provides an isomorphism
\[
\Emp[\eE_S,\eE_T] \cong \bEmp(\El[\eE_S,\eE_T]),
\]
and Proposition~\ref{pro:MApEl=ElMap} gives an isomorphism
\[
\bEmp(\El[\eE_S,\eE_T]) \cong \bEmp[\El(\eE_S),\El(\eE_T)].
\]
Then, \cite[Theorem~5.10(2)]{barbosa2023bundle} yields an isomorphism
\[
\bEmp[\El(\eE_S),\El(\eE_T)] \cong \sDist(N[\El(\eE_S),\El(\eE_T)]).
\]
Finally, Proposition~\ref{pro:sDistNfgMap} provides a surjective map
\[
\sDist(N[\El(\eE_S),\El(\eE_T)]) \to \sDist\left(\Map(N\El(\eE_S),N\El(\eE_T))\right).
\]
Moreover, all of the maps above are convex and preserve non-contextuality in the sense that the non-contextual elements in the codomain are precisely the images of the non-contextual elements in the domain.
\end{proof}

Using the map in Proposition \ref{pro:themapprescontex}, we define the map $\tilde{\mu}_{S,T}$ to be 
$$
({(\zeta' \circ \eta')_T^{-1} \circ -\circ (\zeta' \circ \eta')_S})\circ \mu_{N\El(\eE_S),N\El(\eE_T)} \circ r_{S,T}
$$
Then we have the following commutative 
diagram:
\begin{equation}\label{dia:FTildeF}
\begin{tikzcd}[column sep=huge,row sep=large]
\Emp([S,T],g_{S,T})  \arrow[d,"\tilde{\mu}_{S,T}"] \arrow[rr,"r_{S,T}",tail, twoheadrightarrow]  
&&  \sDist(\Map(N\El(\eE_S),N\El(\eE_T)))  \arrow[d,"\mu_{N\El(\eE_S),N\El(\eE_T)}"] \\
\catConv(\Emp(S),\Emp(T)) \arrow[rr,"{(\zeta' \circ \eta')_T \circ -\circ (\zeta' \circ \eta')_S^{-1}}","\cong"'] && \catConv(\sDist(N\El(\eE_S)),\sDist(N\El(\eE_T))) 
\end{tikzcd}
\end{equation}
Now, suppose we are given a convex map $\tilde{\varphi}\colon \Emp(S) \to \Emp(T)$. By Diagram (\ref{dia:sDistEmp}) the map 
$\tilde{\varphi}$ is induced by a probabilistic procedure if and only if the map 
$$
\varphi=(\zeta' \circ \eta')_T \circ \tilde{\varphi} \circ (\zeta' \circ \eta')_S^{-1} \colon \sDist(N\El(\eE_S)) \to\sDist(N\El(\eE_T))
$$
is a convex combination of maps induced by morphisms in $\catsScen(N\El(\eE_S),N\El(\eE_T))$. By Theorem \ref{thm:contexMapp}, this happens if and only if there exist 
a noncontextual simplicial distribution $p$ on $\Map(N\El(\eE_S),N\El(\eE_T))$ such that $\varphi=\mu_{N\El(\eE_S),N\El(\eE_T)}(p)$. 
Using Proposition \ref{pro:themapprescontex} and Diagram (\ref{dia:FTildeF}), 
this is equivalent to the existence of a noncontextual empirical model $\tilde{p} \in \Emp([S,T],g_{S,T})$ with $\tilde{\varphi}=\mu_{S,T}(\tilde{p})$.

\end{document}